\documentclass[aps,onecolumn,notitlepage]{revtex4-1}

\usepackage{dcolumn}
\usepackage{bm}

\usepackage{amsmath}
\usepackage{graphicx}
\usepackage{amsbsy,amsfonts}
\usepackage{amsthm}
\usepackage{natbib}
\usepackage{colonequals}
\usepackage[colorlinks]{hyperref}
\usepackage{hypcap}
\usepackage{enumitem}

\newcommand{\ket}[1]{\left|#1\right\rangle}

\newcommand{\ketbra}[2]{\left|#1\right\rangle\!\left\langle #2\right|}

\newcommand{\nn}{\nonumber\\}
\newcommand{\Th}[1]{$#1^{\text{th}}$}

\newcommand{\defeq}{\colonequals}
\newcommand{\comment}[1]{\textbf{***[#1]***}}
\newcommand{\suppress}[1]{}

\def\squareforqed{\hbox{\rlap{$\sqcap$}$\sqcup$}}
\def\qed{\ifmmode\squareforqed\else{\unskip\nobreak\hfil
\penalty50\hskip1em\null\nobreak\hfil\squareforqed
\parfillskip=0pt\finalhyphendemerits=0\endgraf}\fi}

\newtheorem{theorem}{Theorem}
\newtheorem{lemma}[theorem]{Lemma}

\newtheorem{definition}{Definition}
\newtheorem{corollary}[theorem]{Corollary}

\newcommand{\eq}[1]{\hyperref[eq:#1]{(\ref*{eq:#1})}}
\renewcommand{\sec}[1]{\hyperref[sec:#1]{Section~\ref*{sec:#1}}}
\newcommand{\app}[1]{\hyperref[app:#1]{Appendix~\ref*{app:#1}}}
\newcommand{\fig}[1]{\hyperref[fig:#1]{Figure~\ref*{fig:#1}}}
\newcommand{\thm}[1]{\hyperref[thm:#1]{Theorem~\ref*{thm:#1}}}
\newcommand{\lem}[1]{\hyperref[lem:#1]{Lemma~\ref*{lem:#1}}}
\newcommand{\cor}[1]{\hyperref[cor:#1]{Corollary~\ref*{cor:#1}}}
\newcommand{\defn}[1]{\hyperref[def:#1]{Definition~\ref*{def:#1}}}

\newcommand{\Rem}{{\mathbf{R}}}
\newcommand{\norm}[1]{\|{#1}\|}
\newcommand{\Norm}[1]{\left\|{#1}\right\|}

\DeclareMathOperator{\Rep}{Rep}

\newcommand{\R}{\mathbb{R}}
\newcommand{\C}{\mathbb{C}}

\newcommand{\ceil}[1]{\lceil{#1}\rceil}

\newcommand{\oddf}[2]{V_{#1,#2}}
\newcommand{\evenf}[2]{W_{#1,#2}}
\newcommand{\bgc}[2]{W^{\rm BGC}_{#1,#2}}
\newcommand{\gen}[2]{\mathcal{F}_{#1}}
\newcommand{\nestgc}[2]{\mathcal{G}_{#1}}
\newcommand{\nestf}[2]{\mathcal{U}_{#1}} 
\newcommand{\JK}{\mathrm{JK}}

\newcommand{\Time}{T}

\newcommand{\A}{\mathcal{A}}
\renewcommand{\H}{\mathcal{H}}

\renewcommand{\eprint}[1]{\href{http://arxiv.org/abs/#1}{#1}}

\renewcommand{\d}{\mathrm{d}}

\begin{document}

\title{Product Formulas for Exponentials of Commutators}
\author{Andrew M. Childs}
\affiliation{Department of Combinatorics \& Optimization, University of Waterloo, Ontario N2L 3G1, Canada}
\affiliation{Institute for Quantum Computing, University of Waterloo, Ontario N2L 3G1, Canada}
\author{Nathan Wiebe}
\affiliation{Department of Combinatorics \& Optimization, University of Waterloo, Ontario N2L 3G1, Canada}
\affiliation{Institute for Quantum Computing, University of Waterloo, Ontario N2L 3G1, Canada}

\begin{abstract}
We provide a recursive method for systematically constructing product formula approximations to exponentials of commutators, giving approximations that are accurate to arbitrarily high order.  Using these formulas, we show how to approximate unitary exponentials of (possibly nested) commutators using exponentials of the elementary operators, and we upper bound the number of elementary exponentials needed to implement the desired operation within a given error tolerance.  By presenting an algorithm for quantum search using evolution according to a commutator, we show that the scaling of the number of exponentials in our product formulas with the evolution time is nearly optimal.  Finally, we discuss applications of our product formulas to quantum control and to implementing anticommutators, providing new methods for simulating many-body interaction Hamiltonians.
\end{abstract}
\maketitle

\section{Introduction}

Product formulas provide a way of approximating a single operator exponential with a product of simpler operator exponentials.  Such formulas are useful in numerical analysis, where they can be applied to the solution of differential equations (see for example \cite{CHMM78}).  More recently, product formulas have become a key tool for quantum simulation \cite{lloyd,AT03,CCDFGS03,Chi04,BACS07,WBHS10,PaZh10,ZAC+10,CK11}.  Hamiltonian simulation using product formulas has numerous applications in quantum information processing, including simulating quantum mechanics \cite{lloyd}, implementing continuous-time quantum algorithms \cite{FGGS00,CCDFGS03,FGG07,HHL09}, and controlling quantum systems (see for example \cite{MB12}).

The primary application of product formulas is to represent exponentials of sums.  Although exponentials of commutators are not as ubiquitous, they arise naturally via their role in Lie groups.  Exponentials of commutators appear in numerous asymptotic expansions, including the Baker--Campbell--Hausdorff series and the Magnus expansion.  They also play a role in quantum computation, such as in the Solovay--Kitaev theorem~\cite{KSV02}, which constructively proves that any finite universal gate set is sufficient to perform efficient universal quantum computation, and in quantum control~\cite{MB12}, where product formulas for exponentials of commutators can be used to suppress couplings or introduce ones that are not naturally present.

Although the theory of product formula approximations for exponentials of sums is well understood, it is considerably less developed in the case of commutator exponentials.  Product formula approximations to commutator exponentials approximate an exponential of the form $\exp([A,B]T)$ for operators $A$ and $B$ and a real number $T$ with a sequence containing exponentials of $A$ and of $B$.  In the limit of small $T$, low-order product formulas for exponentials of commutators are well known.  Methods for systematically constructing high-order product formula approximations to nested commutators have been proposed~\cite{JK97}, although those formulas suffer from numerical stability issues.  More recent work suggests a method for numerically obtaining higher-order product formulas for exponentials of nested commutators~\cite{SL11}.

In this work, we construct arbitrarily high-order product formula approximations to exponentials of commutators that improve upon previous constructions and analyze the performance of these formulas in the context of quantum simulation.  Our formulas are analogous to Suzuki's seminal work on product formulas \cite{suz91}, but apply to the case where the exponentiated operator is a commutator, rather than a sum, of two operators.  We do not explicitly consider cases where the operator is a linear combination of commutators (i.e., a Lie polynomial) or is an ordered commutator exponential, but such cases can be addressed by combining our results with existing product formula approximations for exponentials of sums~\cite{suz91} (see \cite{PaZh10} for an improved analysis in some cases) or ordered operator exponentials~\cite{WBHS10,ZAC+10,PQS+11}.  Furthermore, we present explicit upper bounds for the error in the product formulas and on the number of exponentials needed, and provide a lower bound on the number of exponentials that nearly matches our upper bound.  We also discuss applications of product formulas for exponentials of commutators to quantum simulation.

Our results provide a method to simulate exponentials of the form $e^{[A,B]T}$, for any desired $T \in \R$, using devices that can enact evolution under $A$ or $B$ separately. Specifically, we imagine that we have a pair of devices $\mathcal{D}_A$ and $\mathcal{D}_B$ that take as input an evolution time $t$ and perform operations $e^{At}$ and $e^{Bt}$, respectively. Physically, we can imagine that $\mathcal{D}_A$ and $\mathcal{D}_B$ represent control fields that enact a desired evolution in a quantum system and $t$ represents the time for which those control fields are applied.  Alternatively, we can imagine that $iA$ and $iB$ are Hamiltonians that can be easily simulated and that $\mathcal{D}_A$ and $\mathcal{D}_B$ represent quantum simulation algorithms performing the corresponding evolution. In either case, we measure the efficiency of our formulas by the number of times these devices need to be used to simulate $e^{[A,B]T}$. Note that our figure of merit is \emph{not} the total amount of time the control fields are applied: an elementary evolution $e^{At}$ or $e^{Bt}$ has unit cost independent of $t$.

The remainder of the paper is organized as follows.  In \sec{expcomm} we present basic product formulas for approximating exponentials of commutators. \sec{nested} considers approximating nested commutators, with a scheme using the formulas of \sec{expcomm} presented in \sec{ournest} and a scheme using generic formulas presented in \sec{genericnest}. Upper bounds for the approximation errors are derived in~\sec{errbd}. These error bounds are applied in \sec{lintime}, where we show that the resulting formulas use a number of exponentials that is only slightly superlinear in the evolution time.  We show that this performance is nearly optimal in \sec{search} by proving that sublinear simulation would violate the quantum lower bound on the query complexity of unstructured search.  We present applications of our techniques in \sec{apps}, including simple examples of quantum control as well as a method for simulating exponentials of anticommutators of operators.  In particular, the simulation of anticommutators provides a novel method for implementing many-body interactions in quantum systems.  We conclude in \sec{conclusion} with a discussion of the results and some open problems.

\section{Basic Product Formulas}\label{sec:expcomm}

In this section we present basic formulas approximating the operator $e^{[A,B]t^{k+1}}$ for small $t$ as a product of powers of $e^{At}$ and $e^{Bt^{k}}$, for a given positive integer $k$. Choosing $k=1$ yields the most efficient formulas for the case where $B$ can be implemented directly.  We consider higher values of $k$ for the case where $B$ is itself a (possibly nested) commutator, as discussed further in \sec{nested}.  We present two recursive constructions, one for $k$ odd and another for $k$ even, giving high-order approximation formulas for $e^{[A,B]t^{k+1}}$ in terms of exponentials of $A$ and $B$.  For every integer $p \ge 1$, we present a formula with approximation error $O(t^{2p+k+1})$ in the limit of small $t$.
 
Since the $k=1$ case is the most natural, we begin with formulas for the case where $k$ is odd in \sec{oddk}.  We then discuss the simpler case where $k$ is even in \sec{evenk}.

\subsection{Odd-$k$ Formulas}\label{sec:oddk}

We now develop a recursive approximation-building method that can be used to construct an arbitrarily high-order approximation to $e^{[A,B]t^{k+1}}$ in terms of a product of powers of $e^{At}$ and $e^{Bt^{k}}$, where $k$ is odd. The construction uses the following initial approximation to the time evolution.

\begin{lemma}\label{lem:u1}
Let $A$ and $B$ be bounded operators, let $k\ge 1$ be a real number, and define
\begin{equation}
\oddf{1}{k}(At,Bt^k) \defeq e^{A t}e^{B t^k}e^{-A t}e^{-B t^k}.
\end{equation}
Then $\oddf{1}{k}(At,Bt^k) = e^{[A,B]t^{k+1}+O(t^{k+2})}$.
\end{lemma}

\begin{proof}
The Baker--Campbell--Hausdorff (BCH) formula implies
\begin{equation}
\oddf{1}{k}(At,Bt^k) = e^{A t+B t^k+\frac{1}{2}[A,B]t^{k+1}+ O(t^{k+2})}e^{-A t-B t^{k}+\frac{1}{2}[A,B]t^{k+1}+ O(t^{k+2})}.
\end{equation}
A second use of the BCH formula gives
\begin{equation}
\oddf{1}{k}(At,Bt^k)= e^{[A,B]t^{k+1}+ O(t^{k+2})},
\end{equation}
which completes the proof.
\end{proof}

The product formula $\oddf{1}{1}(At,Bt) = e^{At} e^{Bt} e^{-At} e^{-Bt}$ is known as the group commutator.  This formula has many applications, including generating optimal control sequences~\cite{MB12} and approximating unitary gates via the Solovay--Kitaev theorem~\cite{KSV02}. We show that higher-order generalizations of this formula can be constructed using an iterative approximation-building method that is reminiscent of Suzuki's method \cite{suz91}.

To use our technique, we must have product formulas for the inverses of our approximations.  Our approximations built from $\oddf{1}{k}(At,Bt^k)$ possess one of two symmetry properties that make their inverses simple to compute.

\begin{definition}\label{def:symmetry}
A product formula $U$ is \emph{symmetric} if $U(X,Y)=U(Y,X)^{-1}$ and is \emph{antisymmetric} if $U(X,Y)=U(-Y,-X)^{-1}$, for all bounded operators $X$ and $Y$.
\end{definition}

In particular, $\oddf{1}{k}$ is symmetric, and we will see that high-order approximations constructed using this product formula are either symmetric or antisymmetric.

Now we are ready to describe the main result of this section, which shows how to construct an arbitrarily high-order approximation.

\begin{theorem}\label{thm:multi}
Let $A$ and $B$ be bounded operators, let $k\ge 1$ be an odd integer, and let $\oddf{p}{k}(At,Bt^k)$ be a product formula with $\oddf{p}{k}(At,Bt^k)=e^{[A,B]t^{k+1}+ O(t^{2p+k})}$ for some positive integer $p$.  Let
\begin{align}
\oddf{p+1}{k}(At,Bt^k)
&\defeq        \oddf{p}{k}(A\gamma_p t,B\big(\gamma_p t\big)^k) 
               \oddf{p}{k}(-A\gamma_p t,-B\big(\gamma_p t\big)^k)\nn
&\qquad \times \oddf{p}{k}(A\beta_pt,B(\beta_pt)^k)^{-1} 
               \oddf{p}{k}(-A\beta_pt,-B(\beta_p t)^k)^{-1} \nn
&\qquad \times \oddf{p}{k}(A\gamma_p t,B(\gamma_pt)^k)
               \oddf{p}{k}(-A\gamma_p t,-B(\gamma_p t)^k)
\label{eq:thmmulti:up+1def}
\end{align}
where
\begin{equation}
\beta_p \defeq (2r_p)^{1/(k+1)}, \qquad 
\gamma_p \defeq (1/4+r_p)^{1/(k+1)}, \qquad
r_p \defeq \frac{2^{\frac{k+1}{2p+k+1}}}{4\left(2-2^{\frac{k+1}{2p+k+1}}\right)}.
\label{eq:beta_gamma_r}
\end{equation}
Then $\oddf{p+1}{1}(At,Bt)=\exp([A,B]t^2 + O(t^{2(p+1)+1}))$ and $\oddf{p+1}{k}(At,Bt^k)=\exp([A,B]t^{k+1}+ O(t^{2(p+1)+k+1}))$ if $k > 1$.  Furthermore, $\oddf{p+1}{k}$ is symmetric if $\oddf{p}{k}$ is antisymmetric and is antisymmetric if $\oddf{p}{k}$ is symmetric.
\end{theorem}

\begin{proof}
The assumption that $\norm{U(t)-\oddf{p}{k}(At,Bt^k)} \in O(t^{2p+k})$, where $U(t) \defeq e^{[A,B]t^{k+1}}$, implies that there exist operators $C(A,B)$, $D(A,B)$, and $E(A,B)$ such that
\begin{equation}
\oddf{p}{k}(At,Bt^k)=U(t)+C(A,B)t^{{2p+k}}+D(A,B)t^{{2p+k+1}}+E(A,B)t^{{2p+k+2}}+O(t^{{2p+k+3}}).\label{eq:thmmulti:Upexpansion}
\end{equation}

Every term in the Taylor series of $\oddf{p}{k}(At,Bt^k)$ is a product of powers of $At$ and $Bt^k$, each of which contributes an odd power of $t$.  Thus each term in $C(A,B)$ and $E(A,B)$ contains an odd total number of $A$ and $B$ operators. Similarly, because $2p+k+1$ is even, each term in $D(A,B)$ contains an even number of $A$ and $B$ operators.  Therefore,
\begin{align}
  C(A,B)&=-C(-A,-B)\nn
  D(A,B)&=D(-A,-B)\nn
  E(A,B)&=-E(-A,-B).\label{eq:thmmulti:symcond}
\end{align}

We use these properties to simplify $\oddf{p}{k}(At,Bt^k) \oddf{p}{k}(-At,-Bt^k)$ for arbitrary $t$.  Specifically, we expand $\oddf{p}{k}$ as a power series and use~\eq{thmmulti:Upexpansion} to show that
\begin{align}
 \oddf{p}{k}(At,Bt^k) \oddf{p}{k}(-At,-Bt^k)
 =U(t)^2
 &+[C(A,B)U(t)-U(t)C(A,B)] t^{2p+k} \nn
 &+[D(A,B)U(t)+U(t)D(A,B)] t^{2p+k+1} \nn
 &+[E(A,B)U(t)-U(t)E(A,B)] t^{2p+k+2}
  +O(t^{2p+k+3})
\label{eq:thmmulti:Upexpand1}.
\end{align}
We then Taylor expand each $U(t)$ (but not $U(t)^2 = U(2^{1/(1+k)} t)$) in~\eq{thmmulti:Upexpand1} to lowest order in $t$ and find that
\begin{align}
 \oddf{p}{k}(At,Bt^k) \oddf{p}{k}(-At,-Bt^k)
 &=U(t)^2+2D(A,B)t^{2p+k+1}+O(t^{2p+k+3})+O(t^{2p+2k+1}).
 \label{eq:thmmulti:Upexpand2}
\end{align}
This implies that \eq{thmmulti:up+1def} has no error terms of order $t^{2p+k}$ because it is a product of three pairs of product formula approximations of the same form as \eq{thmmulti:Upexpand2}.

Next we show that the careful choice of $\beta_p$ and $\gamma_p$ eliminates the term of order $t^{2p+k+1}$. To do so, we relate the error terms of $\oddf{p}{k}(At,Bt^k)$ and $\oddf{p}{k}(At,Bt^k)^{-1}$.  Similarly to \eq{thmmulti:Upexpansion}, we have
\begin{equation}
\oddf{p}{k}(At,Bt^k)^{-1} = U(t)^{-1} + \tilde C(A,B)t^{{2p+k}} + \tilde D(A,B)t^{{2p+k+1}} + \tilde E(A,B)t^{{2p+k+2}}+O(t^{{2p+k+3}}) \label{eq:thmmulti:Upexpand3}
\end{equation}
for some operators $\tilde C(A,B)$, $\tilde D(A,B)$, and $\tilde E(A,B)$.  This expansion directly follows from the symmetry properties of $\oddf{p}{k}$.  For example, if $\oddf{p}{k}$ is symmetric, then $\oddf{p}{k}(At,Bt^k)^{-1}=\oddf{p}{k}(Bt^k,At)$.  We know from the previous discussion that $e^{[B,A]t^{k+1}}=\oddf{p}{k}(Bt^k,At) + O(t^{2p+k+1})=U(t)^{-1}$.
The symmetry of the formula then implies that $U(t)^{-1}=\oddf{p}{k}(At,Bt^k)^{-1}+O(t^{2p+k+1})$, which justifies \eqref{eq:thmmulti:Upexpand3}.   The anti-symmetric case follows similarly.

Since each term in $\tilde C(A,B)$ and $\tilde E(A,B)$ consists of an odd number of $A$ and $B$ operators and each term in $\tilde D(A,B)$ consists of an even number of such operators, we have (similarly to~\eq{thmmulti:symcond})
\begin{align}
  \tilde C(A,B)&=- \tilde C(-A,-B)\nn
  \tilde D(A,B)&=\tilde D(-A,-B)\nn
  \tilde E(A,B)&=-\tilde E(-A,-B).
\label{eq:thmmulti:symcondtilde}
\end{align}
Equations \eq{thmmulti:up+1def}, \eq{thmmulti:Upexpand2}, and \eq{thmmulti:Upexpand3} then imply that
\begin{align}
\oddf{p+1}{k}(At,Bt^k)&=\left(U(\gamma_pt)^2+2D(A,B)(\gamma_pt)^{2p+k+1} \right)\nn
&\quad\times\left(U(\beta_pt)^{-2}+2\tilde D(A,B)(\beta_p t)^{2p+k+1} \right)\nn
&\quad\times\left(U(\gamma_p t)^2+2D(A,B)(\gamma_p t)^{2p+k+1} \right)+O(t^{\min\{{2p+k+3},{2p+2k+1}\}})\nn
&=U(t)+\left[4\gamma_p^{2p+k+1}D(A,B)+2\beta_p^{{2p+k+1}}\tilde D(A,B)\right]t^{2p+k+1}+O(t^{\min\{{2p+k+3},{2p+2k+1} \}}).\label{eq:thmmulti:up+1xpr}
\end{align}  

We can relate $D(A,B)$ to $\tilde D(A,B)$ by noting that
\begin{align}
\openone
&=\oddf{p}{k}(At,Bt^k)\oddf{p}{k}(At,Bt^k)^{-1} \nn
&=(U(t)+C(A,B)t^{2p+k}+D(A,B)t^{2p+k+1}+E(A,B)t^{2p+k+2})\nn
&\quad \times (U(t)^{-1}+\tilde C(A,B)t^{2p+k}+\tilde D(A,B)t^{2p+k+1}+\tilde E(A,B)t^{2p+k+2})+O(t^{\min\{2p+k+3,2p+2k+1\}}).\label{eq:thmmulti:Upexpand4}
\end{align}
By Taylor expanding the resulting formula, we find $\tilde{D}(A,B)=-D(A,B)$ (as well as $\tilde{C}(A,B)=-{C}(A,B)$, and if $k>1$ then $\tilde{E}(A,B)=-{E}(A,B)$).

Recalling the definitions of $\beta_p$ and $\gamma_p$, we can substitute $D(A,B)=-\tilde D(A,B)$ into \eq{thmmulti:up+1xpr} to find
\begin{align}
\oddf{p+1}{k}(At,Bt^k)
&=U(t) + \left[4(1/4+r_p)^{\frac{2p+k+1}{k+1}}-2(2r_p)^{\frac{2p+k+1}{k+1}}\right] D(A,B)t^{{2p+k+1}} + O(t^{\min\{{2p+k+3},{2p+2k+1}\}}). \label{eq:thmmulti:up+1xpr2}
\end{align}
The value of $r_p$ specified in \eq{beta_gamma_r} is a root of the expression in square brackets, so the term of order $t^{2p+k+1}$ vanishes.  This demonstrates that
\begin{align}
\oddf{p+1}{k}(At,Bt^k)
&=e^{[A,B]t^{k+1}}+O(t^{\min\{{2p+k+3},{2p+2k+1}\}})
\label{eq:thmmulti:up+2xpr}
\end{align}
as claimed.

Finally, we must show that $\oddf{p+1}{k}$ is either symmetric or antisymmetric.  First, suppose $\oddf{p}{k}$ is symmetric.  Then
\begin{align}
&\oddf{p+1}{k}(At,Bt^k) \oddf{p+1}{k}(-Bt^k,-At)\nn
&\qquad=\Bigr(\oddf{p}{k}(A\gamma_p t,B(\gamma_p t)^k) \oddf{p}{k}(-A\gamma_pt,-B(\gamma_pt)^k) \oddf{p}{k}(B(\beta_p t)^k,A\beta_p t)\nn
&\qquad\qquad \times \oddf{p}{k}(-B(\beta_p t)^k,-A\beta_p t) \oddf{p}{k}(A\gamma_p t,B(\gamma_p t)^k) \oddf{p}{k}(-A\gamma_p t,-B(\gamma_p t)^k)\Bigr)
\nn&\quad\qquad\times\Bigr(\oddf{p}{k}(-B(\gamma_p t)^k,-A\gamma_p t) \oddf{p}{k}(B(\gamma_p t)^k,A\gamma_p t) \oddf{p}{k}(-A\beta_p t,-B(\beta_p t)^k)\nn
&\qquad\qquad \times \oddf{p}{k}(A\beta_p t,B(\beta_p t)^k) \oddf{p}{k}(-B(\gamma_p t)^k,-A\gamma_p t) \oddf{p}{k}(B(\gamma_p t)^k,A\gamma_p t)\Bigr)=\openone.\label{eq:thmmulti:up+1defexp}
\end{align}
Thus $\oddf{p+1}{k}$ is antisymmetric.  A similar calculation shows that if $\oddf{p}{k}$ is antisymmetric then $\oddf{p+1}{k}$ is symmetric.
\end{proof}

The result of~\thm{multi} shows how to recursively construct the product formula $\oddf{p}{k}$ starting from $\oddf{1}{k}$.  Note that this construction involves terms of the form $\oddf{q}{k}(A\lambda t,B(\lambda t)^k)^{-1}$.  Since $\oddf{q}{k}$ is either symmetric or antisymmetric, its inverse can be represented explicitly using~\defn{symmetry}.

The following corollary shows that symmetrization alone can increase the order of the approximation of $\oddf{p}{1}$ from $O(t^{2p+1})$ to $O(t^{2p+2})$ at the cost of doubling the number of exponentials.

\begin{corollary}\label{cor:multi}
Let $A$ and $B$ be bounded operators, let $\oddf{p}{1}(At,Bt)$ satisfy $\oddf{p}{1}(At,Bt)=e^{[A,B]t^{2}+ O(t^{2p+1})}$, and define
\begin{equation}
  \oddf{p}{1}'(At,Bt)
  \defeq \oddf{p}{1}(At/\sqrt{2},Bt/\sqrt{2}) 
         \oddf{p}{1}(-At/\sqrt{2},-Bt/\sqrt{2}).
\end{equation}
Then
\begin{equation}
  \oddf{p}{1}'(At,Bt)= e^{[A,B]t^2+O(t^{2p+2})}.
\end{equation}
\end{corollary}

\begin{proof}
This is a simple consequence of \eq{thmmulti:Upexpand2} with $k=1$.
Since every term of order $2p+1$ in $\oddf{p}{1}(At,Bt)$ must contain an odd total number of $A$s and $B$s, the terms of order $2p+1$ in $\oddf{p}{1}(At,Bt)$ and $\oddf{p}{1}(-At,-Bt)$ are equal and opposite.  By Taylor expanding $\oddf{p}{1}(At/\sqrt{2},Bt/\sqrt{2}) \oddf{p}{1}(-At/\sqrt{2},-Bt/\sqrt{2})$, it is easy to see that the error terms of order $2p+1$ cancel.  The error in $\oddf{p}{1}'(At,Bt)$ is therefore $O(t^{2p+2})$ as claimed.
\end{proof}

\begin{figure}[t!]
\capstart
\includegraphics{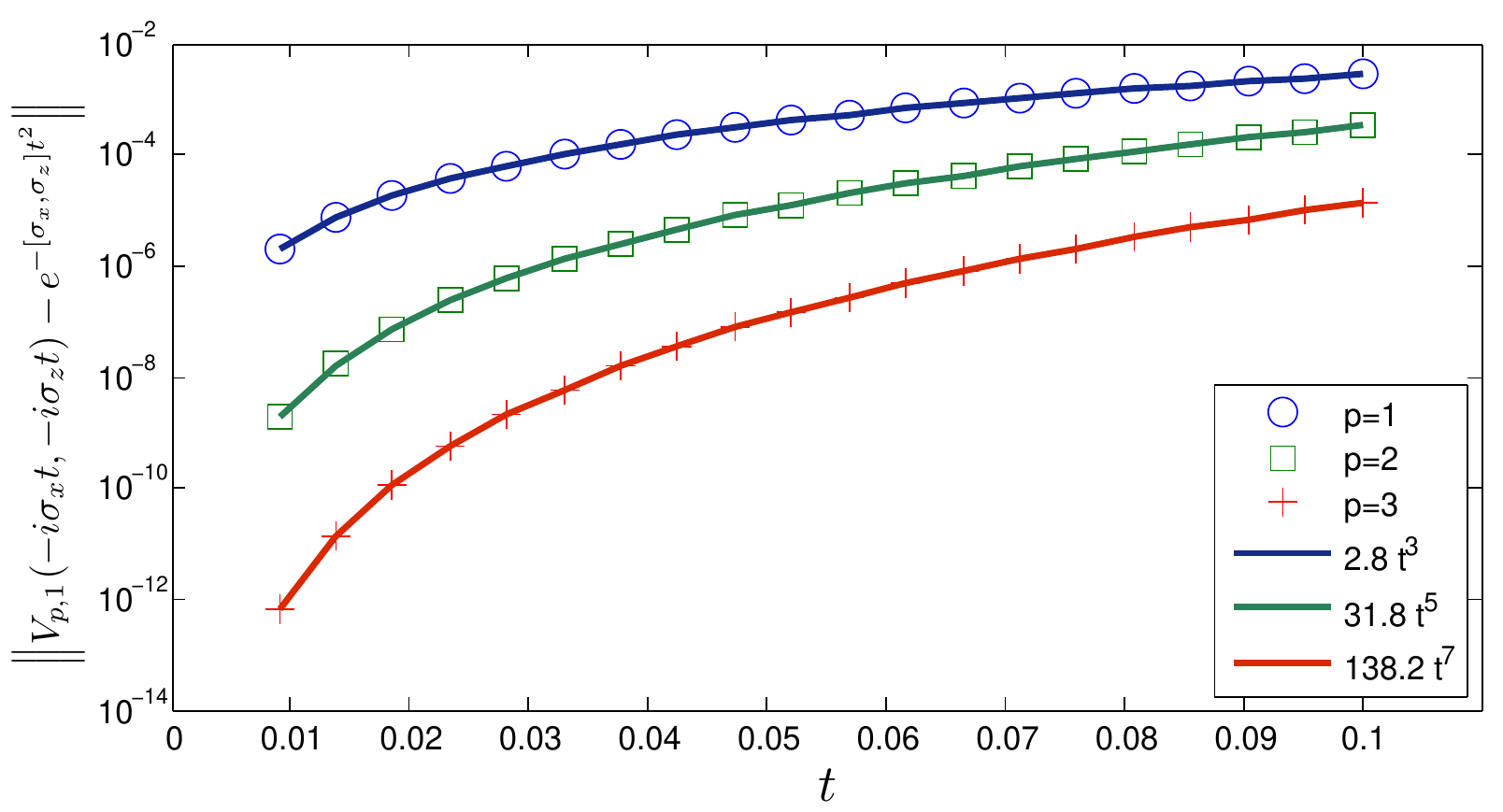}
\caption{Error scaling of the formulas $\oddf{p}{1}(-i\sigma_xt,-i\sigma_zt)$ for $p=1,2,3$ as approximations of $e^{-[\sigma_x,\sigma_z]t^2}$, where $\norm{\cdot}$ is the $2$-norm.
\label{fig:odd_k}}
\end{figure}

Note that a similar symmetrization would improve the error bound of $\oddf{p}{k}$ from $O(t^{2p+k})$ to $O(t^{2p+k+1})$, but this gives no improvement over \thm{multi} for $k>1$ since that theorem already shows that the error is $O(t^{2p+k+1})$.

For $k>1$, we can apply \thm{multi} recursively to produce a formula $\oddf{p}{k}$ with $4 \times 6^{p-1}$ exponentials having error $O(t^{2p+k+1})$.  For $k=1$, applying \thm{multi} recursively followed by one application of \cor{multi} gives a formula $\oddf{p}{1}'$ with $8 \times 6^{p-1}$ exponentials having error $O(t^{2p+2})$.

\fig{odd_k} presents a numerical example showing improved error scaling as $p$ increases.  This example considers simulating the commutator of the Pauli operators $\sigma_x$ and $\sigma_z$ using $k=1$.  These data suggest that our upper bound on the error scaling is in fact tight.

\subsection{Even-$k$ Formulas}\label{sec:evenk}

If $k$ is even then the above approach does not apply, so different reasoning must be used to generate high-order approximations to $e^{[A,B]t^k}$. We show, somewhat surprisingly, that Suzuki's recursive approximation-building formula \cite{suz91} for the exponential of a sum can be used to generate arbitrarily high-order approximation formulas for commutators from a relatively simple initial approximation. The initial formula is as follows.

\begin{lemma}\label{lem:evenk}
Let $A$ and $B$ be bounded operators, let $\xi_k \defeq 2^{-1/(1+k)}$, let $k>0$ be even, and define
\begin{equation}\label{eq:everecur}
\evenf{1}{k}(At,Bt^k) \defeq e^{At\xi_k}e^{Bt^k\xi_k^k}e^{-2At\xi_k}e^{-Bt^{k}\xi_k^k}e^{At\xi_k}.
\end{equation}
Then
\begin{equation}
e^{[A,B]t^{k+1}}=\evenf{1}{k}(At,Bt^k)+O(t^{k+3}).\label{eq:lemevenk:U1expansion}
\end{equation}
\end{lemma}

\begin{proof}
Since $U(t) \defeq e^{[A,B]t^{k+1}}=e^{[B,-A]t^{k+1}}$, it follows from \lem{u1} (which applies for any $k$) that
\begin{align}
e^{[A,B]t^{k+1}}&=\oddf{1}{k}(At\xi_k,Bt^k\xi_k^k)\oddf{1}{k}(Bt^k\xi_k^k,-At\xi_k) + O(t^{k+2})\nn
&=\left(e^{At\xi_k}e^{Bt^k\xi_k^k}e^{-At\xi_k}e^{-Bt^k\xi_k^k}\right)\left( e^{Bt^k\xi_k^k}e^{-At\xi_k}e^{-Bt^k\xi_k^k}e^{At\xi_k}\right)+O(t^{k+2}).\label{eq:lemevenk:U0exp}
\end{align}
The dominant term of equation~\eq{lemevenk:U0exp} reduces to $\evenf{1}{k}(At,Bt^k)$ by simplifying the middle terms.

It remains to show that the term proportional to $t^{k+2}$ vanishes.  This follows from the fact that $\evenf{1}{k}(At,Bt^k)=\evenf{1}{k}(-At,Bt^{k})^{-1}$.  Thus, defining $C(A,B)$ via $\evenf{1}{k}(At,Bt^k) = U(t)+C(A,B)t^{k+2} + O(t^{k+3})$, we have
\begin{equation}
\openone 
= \evenf{1}{k}(At,Bt^k) \evenf{1}{k}(-At,Bt^k)
= \openone+(U(t)C(-A,B)+C(A,B)U(-t))t^{k+2}+O(t^{k+3}).
\label{eq:lemevenk:contradiction}
\end{equation}
Because $k$ is even, $C$ must be composed of an even number of $A$s, so $C(-A,B)=C(A,B)$.  Since~\eq{lemevenk:contradiction} must hold for arbitrary $t$, we have $C(A,B)=0$.
\end{proof}

The product formula in \lem{evenk} is reminiscent of the Strang splitting formula, and indeed reduces to that approximation to $\openone=e^{(A-A)t}$ when  $B=\openone$, or if we take $B=\mathcal{B}^k$ for some operator $\mathcal{B}$ and substitute $k=0$ into the resulting formulas. Since Suzuki's iterative approximation-building method can refine the Strang splitting into arbitrarily high-order formulas, one might suppose that the same recursion could approximate exponentials of commutators. We make this intuition precise in the following theorem.

\begin{theorem}\label{thm:evenk}
Let $A$ and $B$ be bounded operators, let $k > 0$ be an even integer, let $p\ge 1$ be an integer.  Suppose that $\evenf{p}{k}(At,Bt^k)=\evenf{p}{k}(-At,Bt^k)^{-1}$ and $\evenf{p}{k}(At,Bt^k)=e^{[A,B]t^{k+1}+O(t^{2p+k+1})}$, and define
\begin{equation}
  \evenf{p+1}{k}(At,Bt^k)\defeq \evenf{p}{k}(A\nu_p t,B(\nu_p t)^k)^2 \evenf{p}{k}(-A \mu_p t,B(\mu_p t)^k) \evenf{p}{k}(A\nu_p t,B(\nu_p t)^k)^2\label{eq:thmevenk:upplus1def}
\end{equation}
where
\begin{equation}
  \mu_p \defeq (4s_p)^{1/(k+1)}, \qquad
  \nu_p \defeq (1/4+s_p)^{1/(k+1)}, \qquad 
  s_p \defeq \frac{4^{\frac{k+1}{2p+k+1}}}{4\left(4-4^{\frac{k+1}{2p+k+1}}\right)}.
\label{eq:thmevenk:coeffs}
\end{equation}
Then
\begin{equation}
\evenf{p+1}{k}(At,Bt^k) =\exp\left({[A,B]t^{k+1}+O(t^{2p+k+3})}\right).
\end{equation}
\end{theorem}

\begin{proof}
We follow the same reasoning used to analyze the Suzuki formulas.  We have
\begin{equation}
  \evenf{p}{k}(At,Bt^k) = U(t)+C(A,B)t^{k+2p+1}+O(t^{2p+k+2})  \label{eq:thmevenk:updef}
\end{equation}
for some $C(A,B)$. Substituting \eq{thmevenk:updef} into \eq{thmevenk:upplus1def} gives
\begin{align}
\evenf{p+1}{k}(At,Bt^k)=U(t)+ &\left(4C(A,B)\nu_p^{2p+k+1}+C(-A,B)\mu_p^{2p+k+1}\right)t^{2p+k+1}+O(t^{2p+k+2}).
\label{eq:thmevenk:eqn}
\end{align}
Each term comprising $C(A,B)$ must contain an odd number of $A$s, because each $A$ is associated with $t$ and each $B$ is associated with $t^k$, so since $k$ is even, a term proportional to $t^{2p+k+1}$ can only be formed from an odd number of $A$s.  Thus $C(A,B)=-C(-A,B)$.  Therefore, the coefficient of $C(A,B)$ is zero when
\begin{equation}
4(1/4+s_p)^{(2p+k+1)/(k+1)}-(4s_p)^{(2p+k+1)/(k+1)}=0.\label{eq:thmevenk:polyeq}
\end{equation}
The value of $s_p$ specified in \eq{thmevenk:coeffs} is a root of \eq{thmevenk:polyeq} and hence
\begin{equation}
\norm{U(t)-\evenf{p+1}{k}(At,Bt^k)} \in O(t^{2p+k+2}).\label{eq:thmevenk:badscale}
\end{equation}

Finally, we show that by symmetry properties of $\evenf{p+1}{k}$, the actual error scaling is better than in \eq{thmevenk:badscale}. We have
\begin{equation}
  \evenf{p+1}{k}(At,Bt^k)=U(t)+D(A,B)t^{2p+k+2}+O(t^{2p+k+3}).
\end{equation}
Since $\evenf{p}{k}(At,Bt^k)=\evenf{p}{k}(-At,Bt^k)^{-1}$, it is easy to see by multiplication and~\eq{thmevenk:upplus1def} that $\evenf{p+1}{k}(At,Bt^k)=\evenf{p+1}{k}(-At,Bt^k)^{-1}$. We conclude that $D(A,B)=0$ by the same argument used in \eq{lemevenk:contradiction}. Since $D(A,B)=0$, we have $\evenf{p+1}{k}(At,Bt^k)=U(t)+O(t^{2(p+1)+k+1})$ as claimed.
\end{proof}

\begin{figure}[t!]
\capstart
\includegraphics{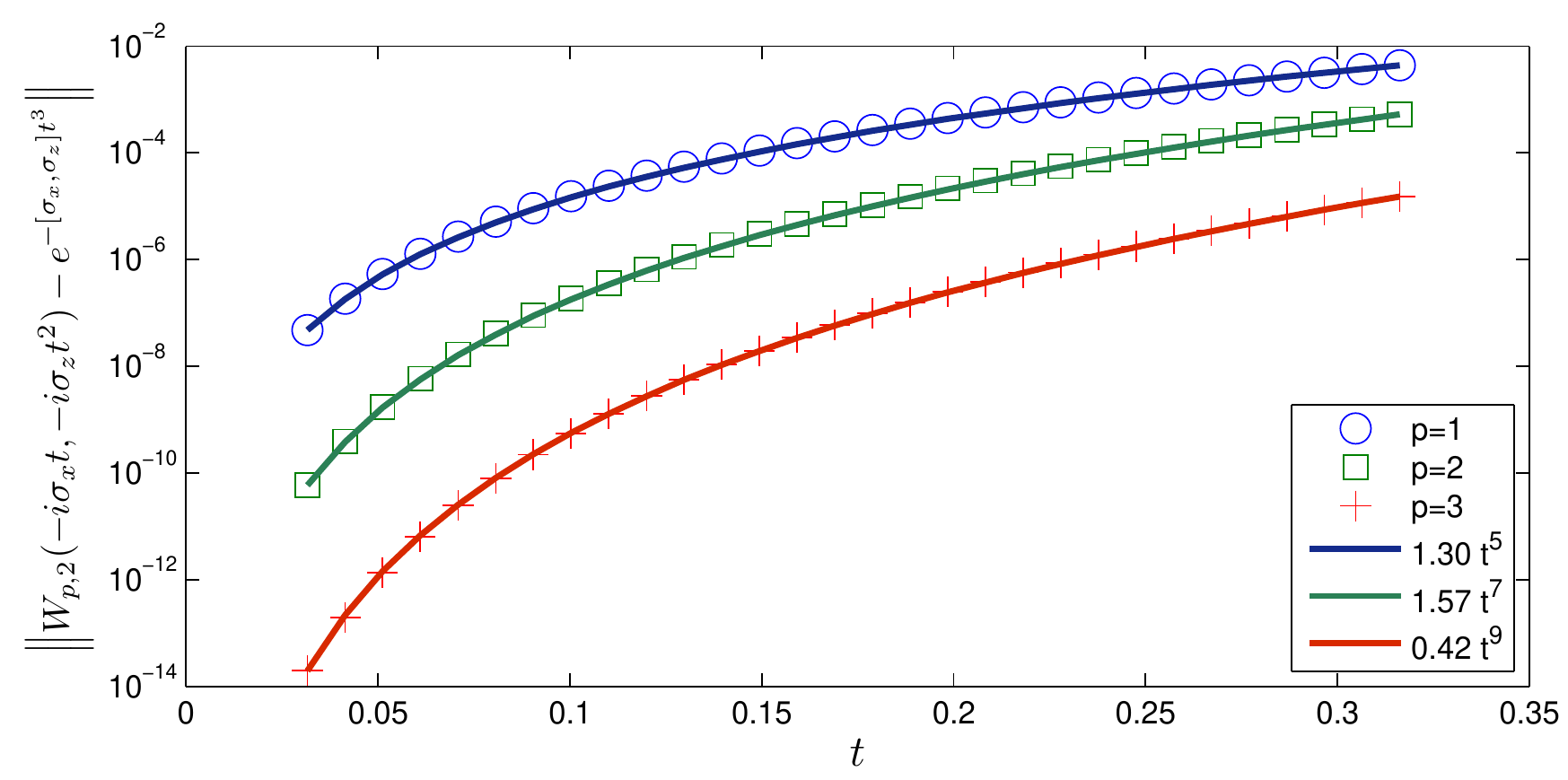}
\caption{Error scaling of the formulas $\evenf{p}{2}(-i\sigma_xt,-i\sigma_zt^2)$ for $p=1,2,3$ as approximations of $e^{-[\sigma_x,\sigma_z]t^3}$.
\label{fig:even_k}}
\end{figure}

The output from the recursive formula in \thm{evenk} can be used as input, so the theorem shows how to refine $\evenf{1}{k}$ into an arbitrarily high-order approximation. The basic formula $\evenf{1}{k}$ has $5$ exponentials, and each iteration increases the number of exponentials by a factor of $5$, so $\evenf{p}{k}(At,Bt^k)$ consists of $5^p$ exponentials.

\fig{even_k} presents a numerical example showing improved error scaling as $p$ increases with $k=2$.  As in the case of odd $k$, these data suggest that our upper bound on the error is tight.

\section{Exponentials of Nested Commutators}\label{sec:nested}

We now construct arbitrarily high-order approximations to exponentials of nested commutators of the form
\begin{equation}
  Z_k
  \defeq [A_k,[A_{k-1},[\ldots,[A_1,A_0]\ldots]]]
  \label{eq:nestedcom}
\end{equation}
(i.e., $Z_0 = A_0$ and $Z_k = [A_k,Z_{k-1}]$ for $k > 0$).
\sec{ournest} presents a construction based on the formulas of \sec{expcomm}, alternating between formulas for odd and even $k$.  \sec{genericnest} presents a method for building a high-order formula from a generic initial formula.

\subsection{Nested Commutators Using the Basic Formulas of \sec{expcomm}}\label{sec:ournest}

Our strategy for approximating such exponentials is simple.  For example, suppose $k$ is odd.  We first approximate $e^{Z_kt^{k+1}}$ with $\oddf{p}{k}(A_k t, Z_{k-1}t^k)$.  The resulting expression still contains commutator exponentials because of the presence of exponentials of $Z_{k-1}$.  We approximate each exponential of the form $e^{Z_{k-1}(\lambda t)^k}$ by $\evenf{p}{k-1}(A_{k-1}\lambda t,Z_{k-2}(\lambda t)^{k-1})$, leaving exponentials of $Z_{k-2}$ in the formula.  This process is repeated recursively until no exponentials of commutators remain.  The case where $k$ is even can be addressed similarly.  In both cases, the resulting product formula is denoted $\nestf{p}{k}(A_kt,\ldots,A_0t)$.  In other words, we approximate a single commutator with $\nestf{p}{1}(A_1 t, A_0 t) \defeq \oddf{p}{1}'(A_1 t, A_0 t)$, and we recursively define $\nestf{p}{q}$ for $q>1$ as
\begin{align}
  \nestf{p}{q}(A_q t, \ldots, A_0 t) \defeq
  \Rep\left(
  \left.\begin{cases}
    \evenf{p}{q}(A_q t, Z_{q-1} t^q)
    & \text{$q$ even} \\
    \oddf{p}{q}(A_q t, Z_{q-1} t^q)
    & \text{$q$ odd}
  \end{cases}\right\}
  ,\, e^{Z_{q-1} (\lambda t)^q} \to \nestf{p}{q-1}(A_{q-1} \lambda t,\ldots, A_0 \lambda t)\right)
\end{align}
for all $q > 1$, where $\Rep(x,\, a \to b)$ indicates that we replace every instance of $a$ in $x$ with $b$.  For example,
\begin{align}
\nestf{1}{2}(A_2t,A_1t,A_0t)
&=\Rep(\evenf{1}{2}(A_2 t, Z_1 t^2),\, e^{Z_1 (\lambda t)^2} \to \nestf{1}{1}(A_1 \lambda t, A_0 \lambda t)) \nn
&=\Rep(e^{A_2\xi_1 t}e^{Z_1 (\xi_1 t)^2}e^{2 A_2\xi_1 t}e^{- Z_1 (\xi_1 t)^2}e^{A_2\xi_1 t},\, e^{Z_1 (\lambda t)^2} \to \nestf{1}{1}(A_1 \lambda t, A_0 \lambda t))\nn
&=e^{A_2\xi_1 t}\nestf{1}{1}( A_{1} \xi_1t, A_0 \xi_1t)e^{2A_2\xi_1 t}\nestf{1}{1}(A_{1}\xi_1 t,A_0 \xi_1 t)^{-1}e^{A_2\xi_1 t}.
\end{align}

At first glance, it may seem surprising that a recursive approximation to $e^{Z_k t^{k+1}}$ can use approximations with error of order less than $2p+k+1$ (rather than using approximations with error $O(t^{2p+k+1})$ at each stage in the recursion) when the desired overall error is $O(t^{2p+k+1})$.  However, since a term with relatively large error is multiplied by a relatively small term, this recursive process straightforwardly yields formulas with error $O(t^{2p+k+1})$. 

\begin{lemma}\label{lem:nested}
For any positive integer $k$, let $A_0,\ldots,A_k$ be bounded operators, define $Z_k$ by equation \eq{nestedcom}, and let $\nestf{p}{k}(A_k t, \ldots, A_0 t)$ be the product formula defined above. Then $\nestf{p}{k}(A_k t, \ldots, A_0 t)=e^{Z_k t^{k+1}+O(t^{2p+k+1})}$.
\end{lemma}

\begin{proof}
We use induction on $k$.  By \cor{multi}, $\nestf{p}{1}(A_1 t,A_0 t)=e^{Z_1t^2 + O(t^{2p+2})}$, establishing the base case. For the induction step, \thm{multi} and \thm{evenk} imply that
\begin{align}
  \nestf{p}{q}(A_q t, \ldots, A_0 t)
  &=\exp([A_q t,Z_{q-1} t^q+O(t^{2p+q})]) \nn
  &=\exp(Z_q t^{q+1} + O(t^{2p+q+1})).
\end{align}
The desired result then follows by induction.
\end{proof}

To calculate the number of exponentials $N_{p,k}$ appearing in $\nestf{p}{k}(A_k t,\ldots,A_0 t)$, recall from \sec{expcomm} that $\oddf{p}{1}'$ uses $N_{p,1}=8 \times 6^{p-1}$ exponentials.  If $k>1$ is odd, then $\oddf{p}{k}$ uses $4 \times 6^{p-1}$ exponentials, half of which require further expansion into $N_{p,k-1}$ exponentials.  If $k$ is even, then $\evenf{p}{k}$ uses $5^p$ exponentials, $2/5$ of which require further expansion into $N_{p,k-1}$ exponentials.  Thus, for $k>1$, we have
\begin{align}
  N_{p,k} &= \begin{cases}
  5^{p-1} (3 + 2 N_{p,k-1}) & \text{$k$ even} \\
  2 \times 6^{p-1} (1 + N_{p,k-1}) & \text{$k$ odd}.
\end{cases}\label{eq:Npk}
\end{align}
While it is cumbersome to present the solution in closed form, it is easy to see that for fixed $k$ we have $N_{p,k}=O(6^{pk})$ since we increase the number of exponentials by a factor of $O(6^p)$ with each iteration.  While this cost may be acceptable for small $k$, it could be prohibitive for large $k$.  We discuss an alternative strategy below that may be more favorable for large $k$.

\subsection{Generic Formulas for Nested Commutators}\label{sec:genericnest}

Another recursive construction of exponentials of nested commutators is presented by Jean and Koseleff~\cite{JK97}.  As in the construction of \sec{ournest}, their approach also recursively refines an initial approximation to a nested commutator exponential into a higher-order formula.  Here we present a modification of their construction with better numerical stability.  Unlike the formulas presented in \lem{nested}, each iteration of the approximation-building algorithm increases the approximation order by one, rather than two.  Thus, for consistency with our previous notation, we denote these product formulas $\gen{p}{2}(A_kt,\ldots,A_0t)$ where $p$ is an integer multiple of $1/2$.  This convention allows us to reuse calculations from the proof of~\thm{evenk}.


\begin{lemma}\label{lem:genericnested}
Let $\gen{p}{2}(A_kt,\ldots,A_0t)$ be an invertible product formula approximation to $\exp(Z_kt^{k+1})$ that has error $O(t^{2p+k+1})$.  Define
\begin{equation}
  \gen{p+1/2}{2}(A_k t,\ldots,A_0 t)\defeq \gen{p}{2}(A_k \nu_p t,\ldots,A_0 \nu_p t)^2 \gen{p}{2}(A_k \mu_p t,\ldots,A_0 \mu_p t)^{-1} \gen{p}{2}(A_k \nu_p t,\ldots,A_0 \nu_p t)^2
\label{eq:bgc0}
\end{equation}
where
\begin{equation}
  \mu_p \defeq (4s_p)^{1/(k+1)}, \qquad
  \nu_p \defeq (1/4+s_p)^{1/(k+1)}, \qquad 
  s_p \defeq \frac{4^{\frac{k+1}{2p+k+1}}}{4\left(4-4^{\frac{k+1}{2p+k+1}}\right)}.
\label{eq:bgc}
\end{equation}
Then $\norm{\gen{p+1/2}{2}(A_kt,\ldots,A_0t)-e^{Z_kt^{k+1}}}\in O(t^{2(p+1/2)+k+1})$.
\end{lemma}

\begin{proof}
We have
\begin{equation}
  \gen{p+1/2}{2}(A_kt,\ldots,A_0t) = \exp(Z_k t^{k+1}) + C(A_k,\ldots,A_0) t^{2p+k+1} + O(t^{2p+k+2})
\end{equation}
for some operator $C(A_k,\ldots,A_0)$.  Similarly, since $\gen{p+1/2}{2}(A_kt,\ldots,A_0t)\gen{p+1/2}{2}(A_kt,\ldots,A_0t)^{-1}=\openone$, we have
\begin{equation}
  \gen{p+1/2}{2}(A_kt,\ldots,A_0t)^{-1} = \exp(-Z_kt^{k+1}) + \tilde{C}(A_k,\ldots,A_0) t^{2p+k+1} + O(t^{2p+k+2})
\end{equation}
for some operator $\tilde{C}(A_k,\ldots,A_0)$. Since
\begin{align}
  \openone
  &= \gen{p+1/2}{2}(A_kt,\ldots,A_0t) \gen{p+1/2}{2}(A_kt,\ldots,A_0t)^{-1} \nn
  &= \openone + (C(A_k,\ldots,A_0)+\tilde C(A_k,\ldots,A_0))t^{2p+k+1} + O(t^{2p+k+2}),
\end{align}
we have $\tilde C(A_k,\ldots,A_0) = -C(A_k,\ldots,A_0)$.  The result then follows from the same calculation as in \eq{thmevenk:eqn}.
\end{proof}

If $\gen{1/2}{k}(A_kt,\ldots,A_0t)$ contains $N_{1/2,k}$ exponentials then $\gen{p}{k}(A_kt,\ldots,A_0t)$ contains $5^{2p-1}N_{1/2,k}$ exponentials.  In contrast, the construction of \cite{JK97} uses only $3^{2p-1}N_{1/2,k}$ exponentials, but the duration of each exponential grows with $p$ and thus the resulting expressions are less numerically stable than those of~\lem{genericnested}.  Indeed, we will see in \sec{lintime} that the formulas of \cite{JK97} are less efficient overall than those produced by \lem{genericnested}.

As a concrete example, suppose we choose $\gen{1/2}{k}(A_kt,\ldots,A_0 t)$ to be the product formula obtained by $k$ recursive applications of the formula $e^{[A,B]t^2}=e^{At}e^{Bt}e^{-At}e^{-Bt}+O(t^3)$ and refine this into a formula that is accurate to $O(t^{2p+k+1})$ by $2p-1$ applications of \lem{genericnested}.  The resulting formula is as follows.

\begin{definition}\label{def:genprod}
Let $\nestgc{p}{k}(A_kt,\ldots,A_0t)$ be the product formula found by $2p-1$ applications of \lem{genericnested} to the product formula $\nestgc{1/2}{k}(A_kt,\ldots,A_0t)$, defined recursively via
\begin{equation}
\nestgc{1/2}{k}(A_k t,\ldots A_0 t) \defeq e^{A_k t}\nestgc{1/2}{k-1}(A_{k-1} t,\ldots A_0 t)e^{-A_k t}\nestgc{1/2}{k-1}(A_{k-1} t,\ldots A_0 t)^{-1},
\end{equation}
where $\nestgc{1/2}{0}(A_0 t) \defeq e^{A_0 t}$.  
\end{definition}

The formula $\nestgc{1/2}{k}(A_kt,\ldots,A_0t)$ contains $N_{1/2,k}\in O(2^k)$ exponentials and is correct to $O(t^{k+2})$.  Thus \lem{genericnested} constructs a product formula with error $O(t^{2p+k+1})$ using $O(5^{2p}2^k)$ exponentials.  We therefore expect that $\nestf{p}{1}(A_kt,\ldots,A_0t)$ will be more efficient than $\nestgc{p}{k}(A_kt,\ldots,A_0t)$ for $k=1$, and for $k\ge 1$ if $p$ is sufficiently small, although the latter formula will be advantageous if both $k$ and $p$ are large.  We make this comparison rigorous in the following section where we provide upper bounds for the approximation errors incurred by these formulas.

\section{Error Bounds for Nested Commutators}\label{sec:errbd}

So far, we have presented product formulas that approximate the evolution according to a (nested) commutator to arbitrarily high order.  We now provide simple upper bounds on the approximation error that results when using these formulas to simulate evolution for a sufficiently small time.  We use the following notation to denote the remainder of a Taylor series expansion truncated at order $\nu-1$.
\begin{definition}
Let $f(t)$ be a bounded operator for any $t\in \mathbb{R}$ with a Taylor series $f(t)=\sum_{n=0}^\infty a_n t^n$.  Then
\begin{equation}
  \Rem_\nu(f(t)) \defeq \sum_{n=\nu}^\infty a_n t^n.
\end{equation}
\end{definition}

We use the following standard result (with a proof included for completeness).

\begin{lemma}\label{lem:expremainder}
Let $\{A_q \colon q=1,\ldots,N\}$ be a set of bounded operators and let $\norm{\cdot}$ be a submultiplicative norm.  Then for any positive integer $\nu$,
\begin{equation}
\Norm{\Rem_\nu \left(\prod_{q=1}^{N}e^{A_q t}\right)}
\le   \Rem_\nu \left( e^{\sum_{q=1}^{N} \norm{A_q} t}\right).
\end{equation}
\end{lemma}

\begin{proof}
Since each $A_q$ is bounded, Taylor's theorem implies that
\begin{equation}
\Norm{\Rem_\nu \left(\prod_{q=1}^{N}e^{A_q t}\right)}
=\Norm{\Rem_\nu \left(\prod_{q=1}^{N}\sum_{p=0}^\infty \frac{(A_q t)^{p}}{p!}\right)}.
\end{equation}
The submultiplicativity of $\norm{\cdot}$ and the triangle inequality imply that
\begin{align}
\Norm{\Rem_\nu \left(\prod_{q=1}^{N}\sum_{p=0}^\infty \frac{(A_q t)^{p}}{p!}\right)}
&\le \Rem_\nu \left(\prod_{q=1}^{N}\Norm{\sum_{p=0}^\infty \frac{(A_q t)^{p}}{p!}}\right) \nn
&\le\Rem_\nu \left(\prod_{q=1}^{N}\sum_{p=0}^\infty \frac{(\norm{A_q} t)^{p}}{p!}\right) \nn
&=\Rem_\nu \left(e^{\sum_{q=1}^{N}\norm{A_q}t}\right)
\end{align}
as claimed.
\end{proof}

The following lemma establishes general error bounds for exponentials of (possibly nested) commutators.

\begin{lemma}\label{lem:errbd}
Let $\{A_j \colon j=0,\ldots,k\}$ be a set of bounded operators and define $Z_k$ by equation \eq{nestedcom}.  Suppose that for some product formula $\gen{p}{k}(A_kt,\ldots A_0t)$,
\begin{enumerate}[itemsep=0pt]
\item\label{item:order} $\gen{p}{k}(A_k t,\ldots,A_0 t) = e^{Z_j t^{k+1} + O(t^\nu)}$ for some integer $\nu>k+1$,
\item\label{item:nq} $\gen{p}{k}(A_k t,\ldots,A_0 t) = \prod_{q=1}^{N_{p,k}} e^{\lambda_q A_{j_q} t}$ where $N_{p,k}^{-1} \sum_{q=1}^{N_{p,k}} |\lambda_q| \le Q_{p,k}$,
\item\label{item:Lambda} $\Lambda \ge 2\norm{A_j}$ for all $j=0,\ldots,k$,
\item\label{item:smalltime} $\Lambda t \le \frac{\ln 2}{N_{p,k} Q_{p,k}}$, and
\item\label{item:nqbig} $N_{p,k} Q_{p,k}\ge 1$.
\end{enumerate}
Then
\begin{equation}
\Norm{e^{Z_kt^{k+1}}-\gen{p}{k}(A_k t, \ldots, A_0 t)} 
\le \left(\frac{e N_{p,k} Q_{p,k} \Lambda  t}{\nu^{1/(k+1)}} \right)^\nu.
\end{equation}
\end{lemma}

\begin{proof}
By assumption \ref{item:order} and the triangle inequality, we have
\begin{equation}
  \norm{e^{Z_kt^{k+1}}-\gen{p}{k}(A_k t,\ldots,A_0 t)}
  \le \norm{\Rem_\nu(e^{Z_kt^{k+1}})} + \norm{\Rem_\nu(\gen{p}{k}(A_k t, \ldots, A_0 t))}.
\label{eq:errbd:split}
\end{equation}
The proof of the lemma follows from upper bounds for both of these terms.  

We first bound $\norm{\Rem_\nu(\gen{p}{k}(A_k t,\ldots,A_0 t))}$.  Assumptions \ref{item:nq} and \ref{item:Lambda}, together with \lem{expremainder}, imply that
\begin{equation}
  \norm{\Rem_\nu(e^{\lambda_1 A_1 t} e^{\lambda_2 A_2 t} \dots)} 
  \le \Rem_\nu\left(e^{\sum_{q=1}^{N_{p,k}} \lambda_q \norm{A_q} t}\right)\le \Rem_\nu(e^{N_{p,k} Q_{p,k}\Lambda t}),
\end{equation}
so Taylor's theorem gives
\begin{align}
\norm{\Rem_\nu(\gen{p}{k}(A_k t,\ldots, A_0 t))}
&\le \frac{(N_{p,k} Q_{p,k}\Lambda t)^\nu}{\nu!} e^{N_{p,k} Q_{p,k}\Lambda t} \nn
&\le \frac{2(N_{p,k} Q_{p,k}\Lambda t)^\nu}{\nu!},
\end{align}
where we have used assumption \ref{item:smalltime} to simplify the bound.
A standard variant of Stirling's formula, namely~\cite{AbSt64}
\begin{align}
  n!&\ge \sqrt{2\pi n} 
  \, n^n e^{-n},  \label{eq:stirling}
\end{align}
gives
\begin{align}
\norm{\Rem_\nu(\gen{p}{k}(A_k t,\ldots, A_0 t))}
&\le \frac{2}{\sqrt{2\pi\nu}}\left(\frac{e N_{p,k} Q_{p,k} \Lambda t}{\nu}\right)^\nu.
\label{eq:nestfbound}
\end{align}

The corresponding bound for the remainder term of $e^{Z_kt^{k+1}}$ follows similarly:
\begin{equation}
\Norm{\Rem_\nu \left(e^{Z_kt^{k+1}}\right)}
\le \frac{(\norm{Z_k}^{1/(k+1)}t)^{(k+1)\ceil{\frac{\nu}{k+1}}}}
         {\ceil{\frac{\nu}{k+1}}!}e^{\norm{Z_k}t^{k+1}}.\label{eq:Rvactual}
\end{equation}
Since $\norm{Z_k} t^{k+1} \le (\Lambda t)^{k+1} \le \ln 2$ (by assumptions \ref{item:Lambda}, \ref{item:smalltime}, and \ref{item:nqbig}), we have
\begin{align}
\Norm{\Rem_\nu \left(e^{Z_kt^{k+1}}\right)}
&\le \frac{2(\Lambda t)^{(k+1)\lceil \frac{\nu}{k+1} \rceil}}{\ceil{\frac{\nu}{k+1}}!} \nn
&\le \frac{2}{\sqrt{2\pi\nu/(k+1)}} \left(\frac{e^{1/(k+1)} \Lambda t}{(\nu/(k+1))^{1/(k+1)}}\right)^\nu
\end{align}
where the second step uses \eq{stirling} and $\Lambda t\le 1$ (by assumptions \ref{item:smalltime} and \ref{item:nqbig}).  To put this in a similar form to \eq{nestfbound}, we use assumption \ref{item:nqbig} to find
\begin{align}
\Norm{\Rem_\nu \left(e^{Z_kt^{k+1}}\right)}
&\le \frac{2}{\sqrt{2\pi\nu/(k+1)}}\left(\frac{k+1}{e^k} \right)^{\nu/(k+1)}  \left(\frac{eN_{p,k}Q_{p,k} \Lambda t}{\nu^{1/(k+1)}}\right)^\nu.
\end{align}
Since $(k+1)/e^k < 1$ for $k \ge 1$ and $\nu/(k+1) > 1$ by assumption \ref{item:order}, we have
\begin{equation}
  \sqrt{k+1} \left( \frac{k+1}{e^k} \right)^{\nu/(k+1)} < \frac{(k+1)^{3/2}}{e^k},
\end{equation}
so
\begin{align}
\Norm{\Rem_\nu \left(e^{Z_kt^{k+1}}\right)}&\le \frac{2}{\sqrt{2\pi\nu}} (k+1)^{3/2}e^{-k} \left(\frac{eN_{p,k}Q_{p,k} \Lambda t}{\nu^{1/(k+1)}}\right)^\nu.
\end{align}
For $k\ge 1$ we have $(k+1)^{3/2}e^{-k}\le 2^{3/2}/e < 1.05$, so
\begin{align}
\norm{e^{Z_kt^{k+1}}-\gen{p}{k}(A_k t,\ldots,A_0 t)}
&\le \frac{2}{\sqrt{2\pi\nu}}\left(\frac{eN_{p,k} Q_{p,k} \Lambda t}{\nu}\right)^\nu
+ \frac{2.1}{\sqrt{2\pi\nu}} \left(\frac{eN_{p,k} Q_{p,k} \Lambda t}{\nu^{1/(k+1)}}\right)^\nu \nn
&\le \frac{4.1}{\sqrt{2\pi\nu}}\left(\frac{eN_{p,k} Q_{p,k} \Lambda t}{\nu^{1/(k+1)}}\right)^\nu
\label{eq:QpUpperbound}
\end{align}
where the second step uses assumption \ref{item:nqbig}.  Since $4.1/\sqrt{2\pi\nu} \le 4.1/\sqrt{6\pi} < 1$, the result follows.
\end{proof}

Note that the bound of \lem{errbd} does \emph{not} directly apply to product formulas that result from using $k>1$ in \thm{multi}, or any of the formulas considered in \thm{evenk}, because such product formulas contain varying powers of $t$.  Error bounds for such cases could be proved by similar reasoning, but are unlikely to give bounds that are as tight, and the cases covered by \lem{errbd} are more relevant for our applications.

The error bound of \lem{errbd} can be used to estimate the maximum size of a time step for product formulas given by \thm{multi} for $k=1$ or any of the product formulas considered in \lem{nested} or~\lem{genericnested}, given upper bounds for $N_{p,k} Q_{p,k}$.  Recall from \sec{ournest} that for $\gen{p}{k}(A_kt\ldots,A_0t)=\nestf{p}{k}(A_kt\ldots,A_0t)$ with constant $k$, $N_{p,k} = O(6^{pk})$.  Furthermore, since $r_p$ and $s_p$ are monotonically decreasing functions of $p$ (for fixed $k$), it can easily be seen that they scale as $1/2+o(1)$ and $1/4+o(1)$, respectively. This similarly implies that $\beta_p=\gamma_p=(1/2)^{1/(k+1)} +o(1)$ and $\mu_p=\nu_p=(1/3)^{1/(k+1)}+o(1)$ for fixed $k$.  Since $Q_{p,k}$ is the maximum value of products of such terms, it follows that $Q_{p,k}\in o(1)$ and thus $N_{p,k}Q_{p,k}\in O(6^{pk})$.  In fact, a more detailed analysis shows that $N_{p,k}Q_{p,k} \in O((2 \sqrt{6})^{pk})$, but this does not significantly improve our results.

We now consider the behavior of $N_{p,k}Q_{p,k}$ for the product formulas constructed in \sec{genericnest}.  From \defn{genprod}, it is easy to see that $N_{p,k} = 5^{2p} N_{1/2,k}$.  If the initial approximation is chosen to be the formula found by recursively applying the group commutator $k$ times, similar to~\cite{JK97}, then
\begin{equation}
N_{1/2,k}=2N_{1/2,k-1} + 2,
\end{equation}
so
\begin{equation}
  N_{1/2,k} = 2^k(\tfrac{1}{2}N_{1/2,1} + 1) - 2. \label{eq:n12}
\end{equation}
Since we have already shown that $Q_{p,k} \in o(1)$,
it follows that $N_{p,k}Q_{p,k}\in O(5^{2p}2^k)$ for the formulas constructed in~\lem{genericnested}.

A direct analytical comparison of the efficiency of the formulas constructed in~\cite{JK97} and those of~\lem{genericnested} is challenging.  It is clear that the construction of of~\cite{JK97} requires a value of $Q_{p,k}\in \Omega(1)$, whereas $Q_{p,k}\in o(1)$ (in fact, $Q_{p,k}\in e^{-\Theta(p)}$) for the method of~\lem{genericnested}.  Since the error bounds provided by~\lem{errbd} scale with $N_{p,k}Q_{p,k}$, a fair comparison requires tight bounds for $Q_{p,k}$.  In fact, these error bounds are likely not tight because they are proved using the triangle inequality.  Thus we cannot rigorously use these error bounds to prove that our formulas are more efficient than those of Jean and Koseleff (denoted $\JK_{p}$ henceforth).  However, we expect that $\nestgc{p}{k}$ should be more efficient than $\JK_{p}$ since the value of $Q_{p,k}$ is significantly smaller in the former case.

\begin{figure}[t!]
\capstart
\includegraphics{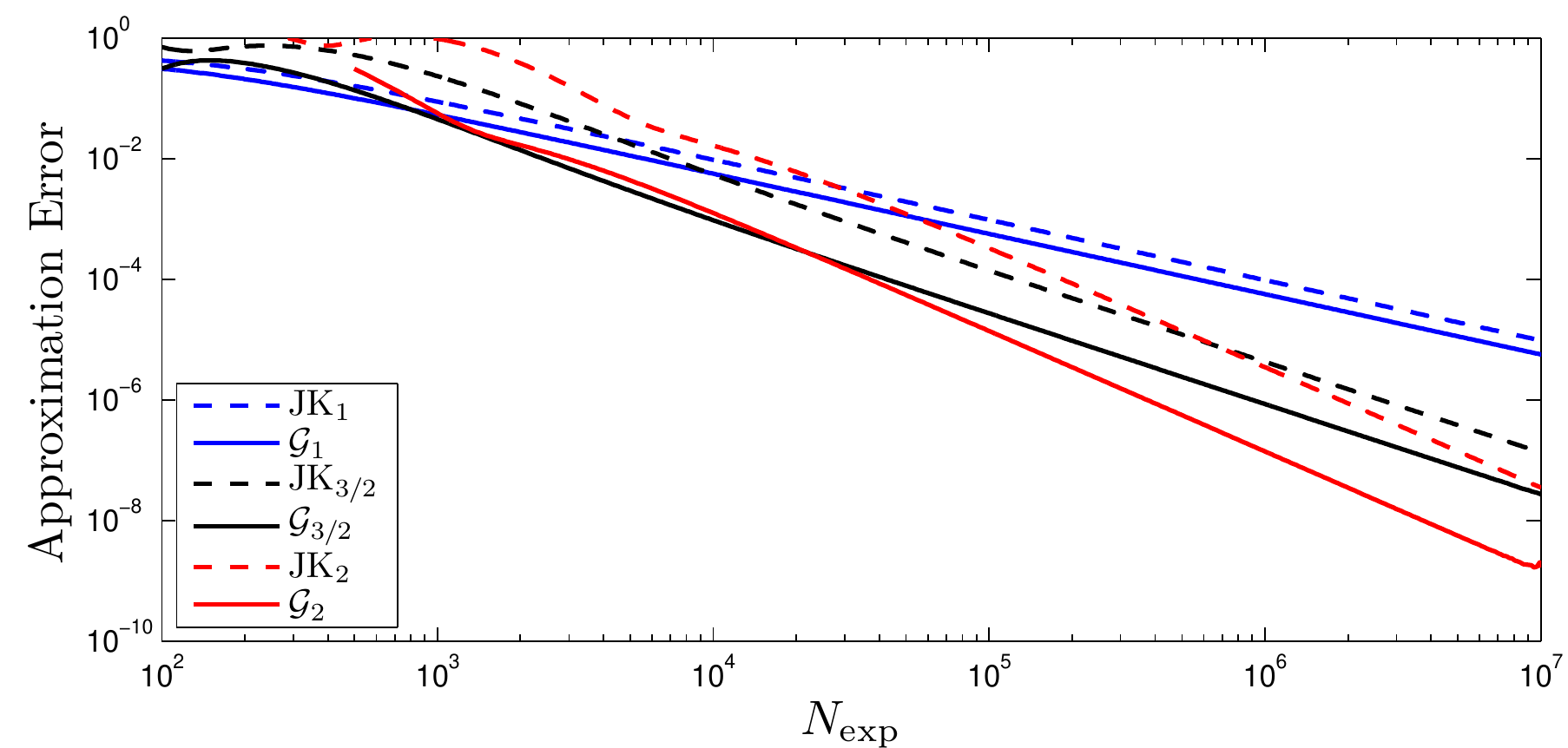}
\caption{Approximation error, as measured by the $2$-norm difference between the approximate and ideal operator exponentials, as a function of the number of exponentials $N_{\exp}$ needed to approximate $e^{[-i \sigma_x,-i\sigma_y]}$ using $\left(\JK_p(-i\sigma_x /\sqrt{r},-i\sigma_y /\sqrt{r})\right)^r$ and $\left(\nestgc{p}{k}(-i\sigma_x /\sqrt{r},-i\sigma_y /\sqrt{r})\right)^r$, where $r$ is the number of time steps used in the approximation.  In both cases, the group commutator is used as the lowest-order approximation.  The data show that our formulas are more efficient than those of~\cite{JK97}.
\label{fig:JKscale}}
\end{figure}

\fig{JKscale} presents numerical evidence for the superiority of $\nestgc{p}{k}$ over $\JK_{p}$.  This figure shows the approximation error of various formulas a function of the number of exponentials used to approximate $e^{-2i\sigma_z}$ by $\left(\JK_p(-i\sigma_x/\sqrt{r},-i\sigma_y/\sqrt{r})\right)^r$ and $\left(\nestgc{p}{k}(-i\sigma_x/\sqrt{r},-i\sigma_y/\sqrt{r})\right)^r$ for $p \in \{1,3/2,2\}$, where $r$ is the number of time steps used.  We see that in all cases that for a fixed number of exponentials, our formulas provide greater accuracy than the formulas of~\cite{JK97}.  We also see evidence of slower convergence to $O(t^{2p+k+1})$ error scaling for $\JK_p$, as expected due to the large value of $Q_{p,k}$ for these formulas.  In addition to improved numerical efficiency, we expect that the formulas given by~\lem{genericnested} are more numerically stable than those of~\cite{JK97} since the backward time steps used in our construction are much shorter.  For these reasons, we suspect that the method of \lem{genericnested} will typically be preferable to that of~\cite{JK97} in both numerical analysis as well as applications to quantum information processing.

\section{Near-Linear Scaling With Total Evolution Time}\label{sec:lintime}

\lem{errbd} bounds the distance between the true evolution $U(t)$ and the approximation $\nestf{p}{k}(A_k t,\ldots,A_0 t)$.  Although the formulas are very accurate for small $t$, the error (as measured by a sub-multiplicative norm such as the spectral norm) can be large if $\Lambda t>1$.  To construct a product formula approximation that is close to $U(t)$ for general $\Lambda t$, we divide the evolution into short segments and show that the cumulative error from these segments is at most $\epsilon$.

To ensure additivity of approximation errors from each time step, we restrict ourselves to the case where the operators $A_j$ are anti-Hermitian (as in quantum mechanics), a restriction that was unnecessary in previous sections.  By carefully choosing $p$ to be a function of $t$ and $\epsilon$, the number of exponentials needed to approximate $U(t)$ within error $\epsilon$ scales only slightly superlinearly with $t^{k+1}$ (the total evolution time) and sub-polynomially with $1/\epsilon$.

The following lemma is the main result of this section.  It provides an upper bound for the minimum number of time steps, $r$, needed to ensure that the cumulative approximation errors from all $r$ time steps sum to at most $\epsilon$.

\begin{lemma}\label{lem:r}
Let $\{A_j \colon j=1,\ldots,k\}$ be a set of bounded anti-Hermitian operators, let $\gen{p}{k}(A_k t,\ldots,A_0 t)$ be a product formula approximation to $e^{Z_k t^{k+1}}$, and let $\epsilon>0$.  If the 
assumptions of \lem{errbd} are satisfied for the product formula $\gen{p}{k}(A_k t/r^{1/(k+1)},\ldots,A_0 t/r^{1/(k+1)})$, where the integer $r$ satisfies
\begin{equation}
r\ge \frac{\left(\frac{eN_{p,k}Q_{p,k}\Lambda t}{(2p+k+1)^{1/(k+1)}} \right)^{k+1+\frac{(k+1)^2}{2p}}}{\epsilon^{(k+1)/(2p)}},\label{eq:rbd}
\end{equation}
then $\norm{U(t)-\gen{p}{k}(A_k t/r^{1/(k+1)},\ldots,A_0 t/r^{1/(k+1)})^r}\le \epsilon$.
\end{lemma}

\begin{proof}
Since the operators $A_j$ are anti-Hermitian, $e^{A_jt}$ and $e^{Z_kt^{k+1}}$ are unitary.  Errors in unitary operations are subadditive, so the total error is at most the number of time steps used in the simulation, $r$, times the error in simulating the commutator exponential during each time step.  
With $r$ time steps, the duration of each step is $t/r^{1/(k+1)}$.
By \lem{errbd}, the error in approximating $e^{Z_kt^{k+1}}$ is at most $\epsilon$ if
\begin{equation}
\left(\frac{eN_{p,k}Q_{p,k}\Lambda t}{(2p+k+1)^{1/(k+1)}r^{1/(k+1)}} \right)^{2p+k+1}\le \frac{\epsilon}{r}\label{eq:rbd2},
\end{equation}
i.e., if
\begin{equation}
\frac{1}{\epsilon}\left(\frac{eN_{p,k}Q_{p,k}\Lambda t}{(2p+k+1)^{1/(k+1)}} \right)^{2p+k+1}\le r^{(2p)/(k+1)}\label{eq:rbd3},
\end{equation}
which is equivalent to the stated condition \eq{rbd} on $r$.
\end{proof}

We can simplify the assumptions needed for \lem{r} 
for certain formulas
because some of the assumptions are automatically satisfied if $\epsilon$ is sufficiently small.  We formalize this in the following corollary.
\begin{corollary}\label{cor:assumptions}
Let $\{A_j \colon j=1,\ldots,k\}$ be a set of bounded anti-Hermitian operators and let $\epsilon>0$.  If Assumptions  \ref{item:order}, \ref{item:nq}, and \ref{item:Lambda} 
 of \lem{errbd} are satisfied for $\nestf{p}{k}(A_kt/r^{1/(k+1)},\ldots,A_0t/r^{1/(k+1)})$ or $\nestgc{p}{k}(A_k t/r^{1/(k+1)},\ldots,A_0 t/r^{1/(k+1)})$, where the integer $r$ satisfies~\eq{rbd} and the error tolerance $\epsilon$ satisfies
\begin{equation}\epsilon \le \left(\frac{e}{(2p+k+1)^{1/(k+1)}} \right)^{2p+k+1}(\ln 2)^{2p}\left(N_{p,k}Q_{p,k}\Lambda t \right)^{k+1},\label{eq:epsrestrict}
\end{equation}
then $\norm{U(t)-\nestf{p}{k}(A_k t/r^{1/(k+1)},\ldots,A_0 t/r^{1/(k+1)})^r} \le \epsilon$ or $\norm{U(t)-\nestgc{p}{k}(A_k t/r^{1/(k+1)},\ldots,A_0 t/r^{1/(k+1)})^r} \le \epsilon$, respectively.
\end{corollary}
\begin{proof}
The argument that the restriction on $\epsilon$ in \eq{epsrestrict} implies assumptions \ref{item:smalltime} and \ref{item:nqbig} of \lem{errbd} is straightforward.  We first demonstrate that $N_{p,k}Q_{p,k}\ge 1$ for $\nestf{p}{k}(A_k t/r^{1/(k+1)},\ldots,A_1 t/r^{1/(k+1)})$ and then use this fact to show the validity of assumption~\ref{item:smalltime}.  We show this by proving loose lower bounds for $N_{p,k}$ and $Q_{p,k}$ and showing that the product of these lower bounds yields a result greater than $1$.  First, we use~\eq{Npk} to see that
\begin{align}
N_{p,k}&\ge 2\cdot 5^{p-1}N_{p,k-1}\ge (2\cdot 5^{p-1})^{k},\label{eq:Npkbd}
\end{align}
where the last inequality follows from solving the prior recursion relation using $N_{p,1}\ge 2\cdot 5^{p-1}$ as the initial condition. 

Since $Q_{p,k}\ge \max_j\{|\lambda_j|\}$, where $\lambda_j$ is the duration of the \Th{j} exponential in $\nestf{p}{k}(A_k t/r^{1/(k+1)},\ldots,A_1 t/r^{1/(k+1)})$ divided by $t/r^{1/(k+1)}$, a lower bound on $|\lambda_j|$ is also a lower bound for $Q_{p,k}$.  The form of $\nestf{p}{k}(A_k t/r^{1/(k+1)},\ldots,A_1 t/r^{1/(k+1)})$ requires that we recursively use the results of~\thm{multi} and~\thm{evenk}.  Note that this recursion must be used $k$ times and each such formula is constructed by $p$ applications of the recursive relations in \thm{multi} or \thm{evenk}.  \thm{multi} and \thm{evenk} imply that each $\lambda_j$ must be a product of at most $pk$ terms of the form $\xi_k$, $\mu_\ell$, $\nu_\ell$, $\beta_\ell$, and $\gamma_\ell$ for $\ell \in \{1,\ldots,p\}$.  Because $\mu_\ell$, $\nu_\ell$, $\beta_\ell$, and $\gamma_\ell$ are monotonically decreasing functions of $\ell$, taking their limit as $\ell\rightarrow \infty$ shows that $\max\{\mu_\ell,\nu_\ell,\beta_\ell,\gamma_\ell\}\ge 1/3^{1/(k+1)}\ge1/\sqrt{3}$.  Furthermore, $\xi_k= 2^{-1/(k+1)}\ge 1/\sqrt{3}$.
Because $\lambda_j$ is a product of $pk$ such terms, it follows that
\begin{equation}
Q_{p,k}\ge \min_j|\lambda_j|\ge 3^{-pk/2}.\label{eq:Qpkbd}
\end{equation}
Combining \eq{Npkbd} and \eq{Qpkbd} gives
\begin{equation}
N_{p,k}Q_{p,k}\ge \frac{(2\cdot 5^{p-1})^k}{3^{pk/2}}\ge \frac{2}{\sqrt{3}}>1.
\end{equation}

For the product formulas generated using~\lem{genericnested}, we have $Q_{p,k}\ge 3^{-p}$.  This is because $\max\{\mu_\ell,\nu_\ell,\beta_\ell,\gamma_\ell\}\ge 1/3^{1/(k+1)}\ge1/\sqrt{3}$ for such formulas since the recursion relation used to generate the formulas is the same as in~\thm{evenk}.  This recursion relation is applied $2p$ times to build a formula with error $O(t^{2p+k+1})$, so $Q_{p,k}\ge 3^{-p}$.  Furthermore, we know that $N_{p,k} \ge 5^{2p-1} 2^k(\tfrac{1}{2}N_{1/2,1} + 1) - 2$ from~\eq{n12}.
It then can be seen that $N_{1/2,1}\ge 4$ because the Baker--Campbell Hausdorff formula can be used to show that if $N_{1/2,1}\le 3$ then no product formula for $e^{[A,B]t^{2}}$ exists with error that scales as $O(t^{3})$. Therefore,
\begin{equation}
N_{p,k} \ge 5^{2p-1} 2^k(3) - 2\ge 5^{2p-1}2^k
\end{equation}
 and hence
\begin{equation}
N_{p,k}Q_{p,k}\ge 5^{2p}2^k3^{-p} \ge 1.
\end{equation}
Therefore $N_{p,k}Q_{p,k}\ge1$ is satisfied for both formulas.

Since $N_{p,k}Q_{p,k}\ge 1$, assumption~\ref{item:smalltime} of \lem{errbd} is implied by
\begin{equation}
\frac{\Lambda t}{r^{1/(k+1)}}\le \frac{\ln 2}{N_{p,k}Q_{p,k}}.\label{eq:rbd4}
\end{equation}
Using our lower bound on $r$ in~\eq{rbd}, we find that~\eq{rbd4} is implied by
\begin{equation}
\frac{\Lambda t \epsilon^{1/(2p)}}{\left(\frac{eN_{p,k}Q_{p,k}\Lambda t}{(2p+k+1)^{1/(k+1)}} \right)^{1+\frac{k+1}{2p}}}\le\frac{\ln 2}{N_{p,k}Q_{p,k}},\label{eq:rbd5}
\end{equation}
which is equivalent to~\eq{epsrestrict} after solving for $\epsilon$ and simplifying the resulting expression.
\end{proof}

Rather than fixing $p$ and choosing $r$ such that the error is at most $\epsilon$, a more cost-effective strategy is to choose $p$ to minimize the total number of exponentials.
The required number of exponentials to achieve error at most $\epsilon$ is
\begin{equation}
N_{\exp}
=N_{p,k} r
\ge N_{p,k}\frac{\left(\frac{eN_{p,k}Q_{p,k}\Lambda t}{(2p+k+1)^{1/(k+1)}} \right)^{k+1+\frac{(k+1)^2}{2p}}}{\epsilon^{(k+1)/(2p)}}\label{eq:nexp}.
\end{equation}
This expression is minimized for some $p$ due to two competing tendencies: $\epsilon^{-(k+1)/(2p)}$ shrinks as $p$ increases, whereas in both methods $N_{p,k}$ grows exponentially in $p$ for $k$ constant, as discussed in \sec{nested}.  We can approximate the optimal value of $p$ by equating two terms:
\begin{equation}
{N_{p,k}^{k+2}}=\left(\frac{\Lambda t}{\epsilon^{1/(k+1)}}\right)^{(k+1)^2/(2p)}.\label{eq:popt}
\end{equation}
Here we neglect the term $N_{p,k}^{(k+1)^2/(2p)}$ since it is approximately constant for $N_{p,k} \in e^{\Theta(p)}$.  We could also include the detailed behavior of $Q_{p,k}$ to more accurately estimate the optimal value of $p$, but doing so complicates the discussion and does not qualitatively change the result.

Taking logarithms of both sides of \eq{popt}, using $N_{p,k}\in e^{\Theta(p_{\rm opt})}$, and considering constant $k$, we find
\begin{equation}
p_{\rm opt}\in \Theta\left(\sqrt{\log\left(\frac{\Lambda t}{\epsilon}\right)} \right).
\end{equation}
Since $p_{\rm opt}$ is sublogarithmic in $\Lambda t/\epsilon$, we find that $N_{p,k} Q_{p,k}$ grows subpolynomially.  Thus, by choosing $p=p_{\rm opt}$, the number of exponentials used to approximate $e^{Z_kt^{k+1}}$ scales as
\begin{equation}
N_{\exp} \in (\Lambda t)^{k+1} \left(\frac{\Lambda t}{\epsilon}\right)^{o(1)}.\label{eq:linscale}
\end{equation}

Similarly to the simulation of sums of Hamiltonians using Suzuki formulas~\cite{Chi04,BACS07}, this shows that the cost of simulating an operator exponential scales only slightly superlinearly with the total evolution time $t^{k+1}$ and subpolynomially with $1/\epsilon$.  (Note that the analogous scaling for the Suzuki formulas~\cite{suz91} follows upon substituting $k=0$.)

\subsection*{Comparison of Efficiency of Product Formulas}


We now use the results of~\lem{r} and~\cor{assumptions} to compare the complexity of approximating an exponential of a commutator using either $\nestf{p}{k}$ or $\nestgc{p}{k}$ product formulas. We know from~\eq{nexp} that our upper bound for the number of exponentials required for a product formula approximation to $e^{Z_k t^{k+1}}$ depends on $N_{p,k}$ and $Q_{p,k}$.  As discussed previously, $Q_{p,k}\in o(1)$ for both formulas and $N_{p,k}$ scales as $O(6^{pk})$ and $O(5^{2p}2^k)$ for $\nestf{p}{k}$ and $\nestgc{p}{k}$, respectively.  Therefore
\begin{equation}\label{eq:nestfscale}
N_{\exp}\in O\left(\frac{6^{pk}\left(\frac{e6^{pk}\Lambda t}{(2p+k+1)^{1/(k+1)}} \right)^{k+1+(k+1)^2/(2p)}}{\epsilon^{\frac{k+1}{2p}}} \right)
\end{equation}
if $\nestf{p}{k}(A_kt,\ldots,A_0t)$ is used to approximate the commutator exponential.  On the other hand, if $\nestgc{p}{k}(A_kt,\ldots,A_0t)$ is used, then
\begin{equation}\label{eq:nestgcscale}
N_{\exp}\in O\left(\frac{5^{2p}2^k\left(\frac{e5^{2p}2^k\Lambda t}{(2p+k+1)^{1/(k+1)}} \right)^{k+1+(k+1)^2/(2p)}}{\epsilon^{\frac{k+1}{2p}}} \right).
\end{equation}


Equations~\eq{nestfscale} and~\eq{nestgcscale} suggest that product formulas derived from~\lem{nested} will be superior for small $k$ whereas formulas generated from~\lem{genericnested} will be superior for larger values of $k$.  However, we know from~\eq{linscale} that the numbers of exponentials required for both of these formulas are the same up to subpolynomial factors.  This suggests that both methods are roughly comparable in the limit of large $\Lambda t$ and $1/\epsilon$.  In the next section, we show that this performance is nearly optimal, so a substantial improvement would require assumptions about the form of $A_k,\ldots,A_0$.

\section{Optimality Proof Via Quantum Search}\label{sec:search}

We have shown that high-order approximations to exponentials of commutators can be constructed with a number of exponentials that scales almost linearly with the evolution time. Here we show that such schemes are nearly optimal by demonstrating that a simulation using a sublinear number of exponentials would violate the lower bound for quantum search \cite{BBBV97}. The argument is similar to the analysis of the continuous-time algorithm for quantum search \cite{FG96}, but uses a Hamiltonian expressed as a commutator instead of as a sum.

\begin{theorem}
Any product formula approximation for $e^{[A,B]\Time}$ for generic $A,B$ consists of $\Omega(\Time)$ exponentials.
\end{theorem}

\begin{proof}
Consider the problem of searching for an unknown $w \in \{1,\ldots,n\}$ using a black box acting as
\begin{equation}
  \ket{x,b} \mapsto
  \begin{cases}
    \ket{x,b}      & \text{if $x \ne w$} \\
    \ket{x,\bar b} & \text{if $x=w$}.
  \end{cases}
\end{equation}
A quantum computer must make $\Omega(\sqrt{n})$ queries to the black box to solve this problem with bounded error \cite{BBBV97}.

To solve the search problem, we can begin in the state $\ket{+} \defeq \frac{1}{\sqrt{n}}\sum_{x=1}^n \ket{x}$ and try to evolve with the Hamiltonian $i[\ketbra{w}{w},\ketbra{+}{+}]$ for some time $\Time$ such that
\begin{equation}
  e^{[\ketbra{w}{w},\ketbra{+}{+}]\Time}\ket{+}\approx \ket{w}
\end{equation}
for any $w \in \{1,\ldots,n\}$.

The Hamiltonian $i[\ketbra{w}{w},\ketbra{+}{+}]$ generates a rotation in the plane spanned by $\ket{w}$ and $\ket{+}$.  We therefore choose our basis to contain $\ket{w}$ and $\ket{\perp} \defeq \frac{1}{\sqrt{n-1}}\sum_{x \ne w} \ket{x}$. Let
\begin{align}
  Y
  &\defeq \sqrt{n} [\ketbra{w}{w},\ketbra{+}{+}] \nn
  &= \ketbra{w}{+}-\ketbra{+}{w}.
\end{align}
The action of $Y$ on the basis states is
\begin{align}
Y\ket{w}&
         =-\sqrt{\frac{n-1}{n}}\ket{\perp}\\
Y\ket{\perp}&=\sqrt{\frac{n-1}{n}}\ket{w}.
\end{align}
This is simply $i\sqrt{(n-1)/n}$ times the Pauli operator $\sigma_y$ on the relevant subspace, with eigenvectors
\begin{equation}
\ket{e_{\pm}}=\frac{\ket{w} \pm i\ket{\perp}}{\sqrt{2}}
\end{equation}
and corresponding eigenvalues
\begin{equation}
e_{\pm}=\pm i\sqrt{\frac{n-1}{n}}.
\end{equation}
We have
\begin{align}
\ket{w}&=\frac{\ket{e_+}+\ket{e_-}}{\sqrt{2}}\\
\ket{+}&=\frac{\ket{e_+}}{\sqrt{2}}\left(\frac{1}{\sqrt{n}}-i\sqrt{\frac{n-1}{n}} \right)+\frac{\ket{e_-}}{\sqrt{2}}\left(\frac{1}{\sqrt{n}}+i\sqrt{\frac{n-1}{n}} \right).
\end{align}
Defining $\phi$ so that $\tan \phi = \sqrt{n-1}$, the time-evolved state is
\begin{align}
e^{\frac{Yt}{\sqrt{n}}}\ket{+}=\frac{\ket{e_+}}{\sqrt{2}}e^{i\left(\sqrt{\frac{n-1}{n}}\frac{\Time}{\sqrt{n}}-\phi\right)} +\frac{\ket{e_-}}{\sqrt{2}}e^{-i\left(\sqrt{\frac{n-1}{n}}\frac{\Time}{\sqrt{n}}-\phi\right)}.
\end{align}
The evolution then yields the desired state, $\ket{w}$, when the phases of both eigenvectors are the same.  This first occurs when
\begin{equation}
\Time=\frac{\phi n}{\sqrt{n-1}} = \frac{\pi}{2} \sqrt{n} + O(1).
\end{equation}

The exponential of $\ketbra{w}{w}$ for any time can be simulated with $2$ queries to the black box, and the exponential of $\ketbra{+}{+}$ for any time can be simulated without querying the black box.  Thus a product formula approximation of $e^{[\ketbra{w}{w},\ketbra{+}{+}]\Time}$ using $N_{\exp}$ exponentials gives an algorithm for the search problem using $O(N_{\exp})$ queries.  By the $\Omega(\sqrt n)$ lower bound on search, we find that the number of exponentials in the product formula must be $N_{\exp} \in \Omega(\sqrt{n}) = \Omega(\Time)$.
\end{proof}

Thus there is no product formula approximation to $e^{[A,B]t^{k+1}}$, for generic $A$ and $B$, that consists of $o((\Lambda t)^{k+1})$ exponentials.  Equation \eq{linscale} shows that our high-order product formula approximations use $(\Lambda t)^{k+1+o(1)}$ exponentials, so their performance as a function of $\Lambda t$ is nearly optimal.

\section{Applications}
\label{sec:apps}

\subsection{Controlling Quantum Systems}\label{sec:control}

Quantum control is the art of designing sequences of quantum operations (often referred to as pulses) that approximate a desired operation.  This is especially useful when the target operation cannot be directly implemented.  Techniques that are used to design pulse sequences include~\cite{MB12} Lie--Trotter--Suzuki sequences, Dyson series, Fer and Wilcox expansions, Solovay--Kitaev sequences, and the Magnus series.  The latter three techniques explicitly involve implementing exponentials of commutators, which can be turned into pulse sequences using product formulas.  Current approaches to implement such commutator exponentials are typically inaccurate and often must be numerically optimized via a method such as GRAPE \cite{KRK+05,SSK+05}, which uses gradient ascent to numerically optimize the efficiency and robustness of a pulse sequence that implements a desired unitary operation.  The quality of the numerically optimized sequence depends on the initial input sequence, so our results may be useful for designing pulse sequences (and perhaps understanding why a numerically optimized control sequence takes a particular form).

A simple application of our results to quantum control addresses the problem of performing an arbitrary single-qubit rotation given access to two non-orthogonal rotation axes.
Consider the single-qubit Hamiltonian
\begin{equation}
H(t)=B_0 \sigma_z + B_c(t)(\sigma_z+\sigma_x)
\end{equation}
where $B_c(t)$ is a controllable magnetic field that can switch between the values $0$ and $B_c < B_0$.  The goal is to perform the operation
\begin{equation}
U(t)=e^{-i\omega_0\sigma_y t^{2}}
\end{equation}
for some constant $\omega_0$, which is sufficient to establish universal control of the qubit if used in concert with the $z$ rotation provided by the Hamiltonian when $B_c=0$.  We can achieve the desired $y$ rotation as
\begin{equation}
  U(t)=e^{-\omega_0[\sigma_z,\sigma_x]t^{2}/2}=e^{-\omega_0[B_0 \sigma_z, B_0\sigma_z + B_c(\sigma_z+\sigma_x)]t^2/(2B_0B_c)}.
\end{equation}
For simplicity, define
\begin{align}
A&\defeq B_0\sigma_z,\nonumber\\
B&\defeq B_0\sigma_z + \frac{\omega_0}{2B_0}(\sigma_z+\sigma_x),
\label{eq:simpleAB}
\end{align}
which can both be realized by appropriate choices of $B_c$.  Then $[A,B]=i\omega_0\sigma_y$.  \cor{multi} shows that we can approximate $U(t)$ as
\begin{equation}
 U(t)=e^{-iAt/\sqrt{2}}e^{-iBt/\sqrt{2}}e^{iAt/\sqrt{2}}e^{iBt/\sqrt{2}}e^{iAt/\sqrt{2}}e^{iBt/\sqrt{2}}e^{-iAt/\sqrt{2}}e^{-iBt/\sqrt{2}} + O(t^4).\label{eq:genGroupCommutator}
\end{equation}
\thm{multi} can be used to construct higher-order variants of this expression.

Note that we have assumed that we can implement evolutions of the form $\exp(iAt)$ and $\exp(iBt)$ for $t>0$, corresponding to evolution for negative time.  For this single-qubit system, this can easily be done by evolving under $A$ or $B$ for some positive time.  However, this may not be as straightforward in general.

Recent work by Borneman, Granade, and Cory \cite{BGC12} further highlights the role of commutator exponentials in quantum control by using them to suppress undesirable interactions in coupled spin systems and introduce new couplings that are not present in the original Hamiltonian. Their work constructs a nested commutator of the form $e^{i[A,[A, B]]t^{3}}$ using only the ability to implement exponentials of $A$ and $B$ individually.  They achieve this by presenting a product formula for nested commutators that is outside the scope of our formalism (since it exploits the fact that two of the terms in the nested commutator are equal), namely
\begin{equation}
\bgc{1/2}{2}(-iAt,-iBt)\defeq e^{-iAt}e^{-iBt}e^{iAt}e^{iBt}e^{iAt}e^{-iBt}e^{-iAt}e^{iBt}.
\label{eq:bgcdef}
\end{equation}
They show that this product formula obeys
\begin{equation}
\norm{\bgc{1/2}{2}(-iAt,-iBt)-e^{i[A,[A,B]]t^3}} \in O(t^4).\label{eq:bgcscale}
\end{equation}
Here we use the subscript $p=1/2$ since, with $k=2$, the error is $O(t^4)=O(t^{2p+k+1})$ as in our previous product formulas.
The validity of~\eq{bgcscale} follows directly from
\begin{align}
\label{eq:bgcexpand}
\ln(e^{-iAt}e^{-iBt}e^{iAt}e^{iBt})
=-[A,B]t^2+\frac{it^3}{2!}[A+B,[A,B]]+\frac{t^4}{3!}([A,[B,[B,A]]]/2+[A+B,[A+B,[A,B]]])+O(t^5)
\end{align}
and the Baker--Campbell--Hausdorff formula: all terms of order $t^3$ and higher cancel in \eq{bgcdef} except for $it^3[A,[A,B]]$.

Note that $\bgc{1/2}{2}(-iAt,-iBt)$ is actually correct to $O(t^5)$ if $[A,[B,[B,A]]]=0$, which holds, for example, when $A$ and $B$ are Pauli operators.  We refer to the BGC formula as $\bgc{1}{2}(-iAt,-iBt)$ in such cases to reflect the fact that the error is $O(t^5)$.  Also note that by a similar calculation to \eq{bgcexpand},
\begin{equation}
\bgc{1/2}{2}(-iAt,-iBt)^{-1}= e^{-iBt}e^{iAt}e^{iBt}e^{-iAt}e^{-iBt}e^{-iAt}e^{iBt}e^{iAt}
\label{eq:bgcinv}
\end{equation}
satisfies 
\begin{equation}
\norm{\bgc{1/2}{2}(-iAt,-iBt)^{-1}-e^{-i[A,[A,B]]t^3}} \in O(t^4).\label{eq:bgcscaleinv}
\end{equation}
This error term is also $O(t^5)$ when $[A,[B,[A,B]]]=0$.

Even though the BGC formula does not have the symmetries required by \thm{multi} or \thm{evenk}, we can still use our results to refine it into higher-order versions.  Specifically, \lem{genericnested} shows that we can refine it into a product formula with error $O(t^{2p+k+1})$.  The construction described in \lem{genericnested} refines $\bgc{1/2}{2}(-iAt,-iBt)$ into a formula denoted $\bgc{1}{2}(-iAt,-iBt)$ that has error $O(t^{5})$ (rather than $O(t^{6})$ for the reasons described above).  We can then continue to refine the formula to arbitrarily high order by applying the construction recursively.



An example of this generalized BGC formula with $A=i\sigma_x$ and $B=i\sigma_z$ is presented in \fig{bgc}. Here the lowest-order formula is $\bgc{1}{2}(-i\sigma_xt,-i\sigma_zt)$ since $[A,[B,[B,A]]]=0$.  We see that the upper bound of \lem{genericnested} is essentially optimal: the approximation error scales as $O(t^{2p+3})$.  Note that these higher-order BGC formulas use $O(5^{2p})$ exponentials, an improvement over the $O(6^{2p})$ exponentials that would result from applying \lem{nested}.  

\begin{figure}[t!]
\capstart
\includegraphics{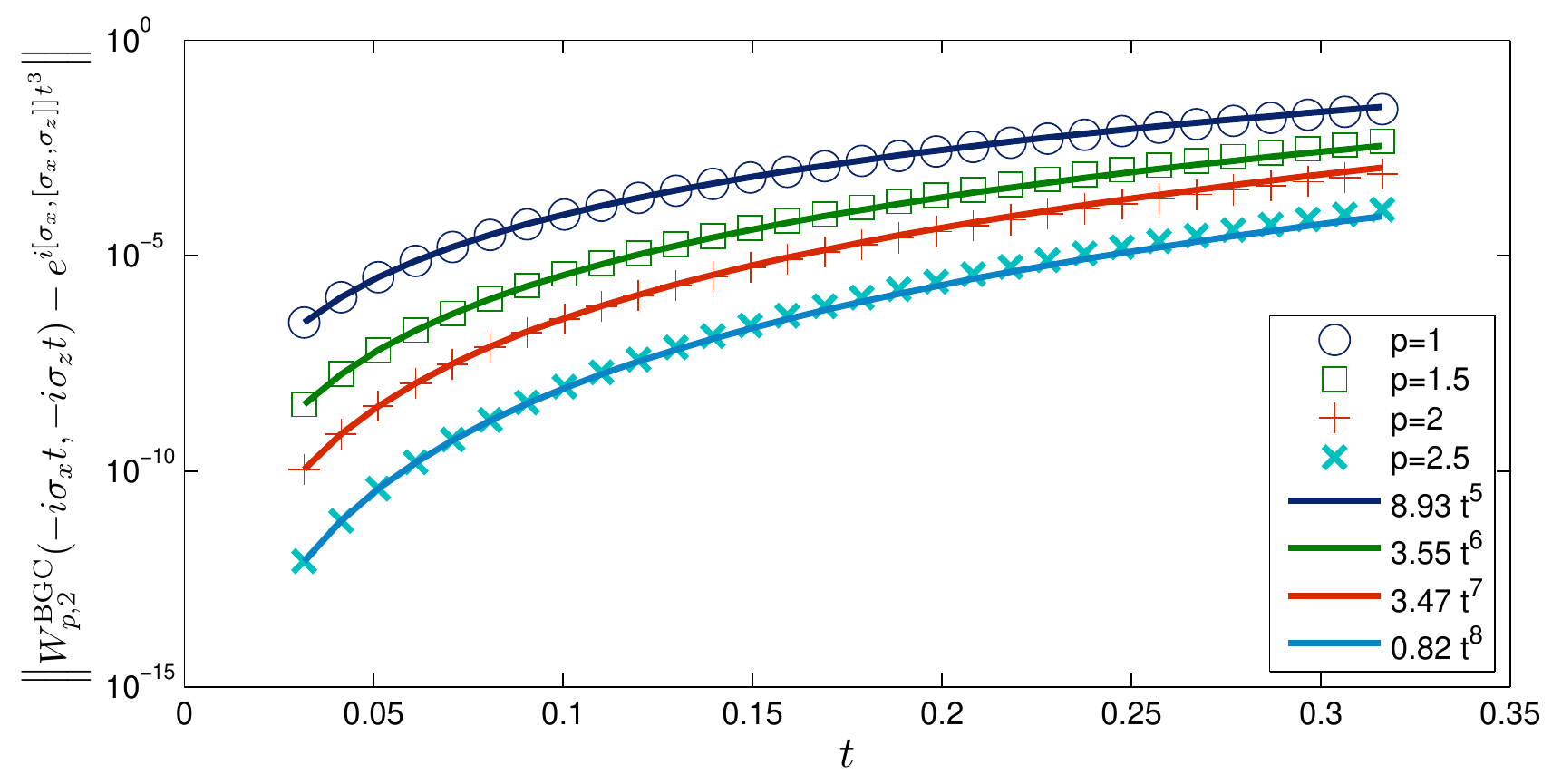}
\caption{Error scaling of the formulas $\bgc{p}{2}(-i\sigma_xt,-i\sigma_zt)$ for $p=1,1.5,2,2.5$ as approximations of $e^{i[\sigma_x,[\sigma_x,\sigma_z]]t^2}$, where $\norm{\cdot}$ is the $2$-norm.
\label{fig:bgc}}
\end{figure}

As another simple application, we could also use the nested commutator formulas provided by~\lem{nested} or \lem{genericnested} to implement $e^{[A,[A,B]]t^3}$ in the example discussed at the beginning of this section, where $A$ and $B$ are given by \eq{simpleAB}.  This produces an effective $x$ rotation, effectively canceling the $\sigma_z$ terms present in $B$.

\subsection{Product Formula Approximations for Anticommutators}

Another application of our results is to implement Hamiltonian evolution according to anticommutators.  The anticommutator of two operators $A$ and $B$ is
\begin{equation}
\{A,B\} \defeq AB+BA.
\end{equation}
We can implement $e^{\{A,B\}t^2}$ using a commutator simulation by enlarging the Hilbert space.  This method can also be used to implement powers of a Hamiltonian, or more generally, products of powers of commuting Hamiltonians, which may be applied to simulate many-body couplings.  In contrast to the techniques described in the previous section, this method can be used to introduce interactions that do not arise naturally as commutators of terms already appearing in the Hamiltonian.  Our product formula approach for simulating anticommutators can therefore be seen as a complementary method to the generalization of the BGC approach described in~\sec{control}.

We extend the Hilbert space from $\H$ to $\H \otimes \C^2$ and create analogs of our operators $A$ and $B$ on this enlarged space, namely
\begin{align}
A' &\defeq A\otimes \sigma_y\nn
B' &\defeq B\otimes \sigma_x,\label{eq:dilation}
\end{align}
where $\sigma_x$ and $\sigma_y$ are Pauli operators.  Then
\begin{align}
[A',B']
&=AB\otimes\sigma_y\sigma_x-BA\otimes\sigma_x\sigma_y \nn
&=-i\{A,B\}\otimes\sigma_z.
\end{align}

Let $\ket{\psi} \in \H$ be any quantum state.  Then
\begin{align}
[A',B']\ket{\psi}\otimes\ket{0}
&= -i\{A,B\}\ket{\psi}\otimes\ket{0},\label{eq:comm2}
\end{align}
so the action of the anticomutator can be simulated using an ancilla qubit in the eigenstate $\ket{0}$ of $\sigma_z$.  In particular,
\begin{align}
e^{-i\{A,B\}t^2}\ket{\psi}\otimes\ket{0}
&=e^{[A',B']t^2}\ket{\psi}\otimes\ket{0} \nn
&=\oddf{p}{1}'(A't,B't)\ket{\psi}\otimes\ket{0}+O(t^{2p+2}).
\end{align}
This shows how to simulate anticommutators using a simulation of commutators. It is straightforward to generalize this construction to simulate nested anticommutators in terms of nested commutators.

While we have already discussed applications of implementing exponentials of commutators, exponentials of anticommutators may seem less natural.  However, one application of anticommutators (between commuting operators) is to simulate many-body couplings.  For example, consider the one-body operators $A_j=\sigma_z^{(j)}$ for $j = 0,\ldots,k$, where the superscript indicates which spin is acted on.  Then
\begin{equation}
\{A_k, \ldots \{A_2,\{A_1,A_0\}\} \ldots \}=2^k \sigma_z^{\otimes k+1},\label{eq:sigmazconst}
\end{equation}
so the Hamiltonian $\sigma_z^{\otimes {k+1}}$ can be simulated using a product of exponentials that each involve only two-body terms, which are typically more natural than $k$-body terms with $k>2$.  This approach provides an alternative to direct simulation methods for many-body terms \cite{NC00,LWG+10,BMS+11,RWS12}.  Since no explicit changes of basis are performed, this simulation may more accurately reproduce features of the ideal evolution (e.g., its behavior in the presence of errors).  

As an example of a complete simulation, consider the toric code Hamiltonian \cite{kit03}
\begin{equation}
H= -J\left(\sum_{v}A_v + \sum_p B_p\right),
\end{equation}
where $A_v \defeq \prod_{i\in v}\sigma_x^{(i)}$ and $B_p \defeq \prod_{j \in p} \sigma_z^{(j)}$ are four-body coupling terms acting on qubits on the edges of a square lattice, where $v$ denotes the set of qubits surrounding a vertex of the lattice and $p$ denotes the set of qubits on the edges of a plaquette.  The vertex and plaquette operators commute, so the time-evolution operator is
\begin{equation}
e^{-iHt}=\prod_{v} e^{iJA_vt}\prod_{ p} e^{iJB_pt}.\label{eq:toricham}
\end{equation}
Each of these exponentials involves a four-body Hamiltonian and therefore can be implemented using only two-body interactions via the approach suggested by~\eqref{eq:sigmazconst} with $k=3$.

The following theorem generalizes the above examples to simulate a Hamiltonian that is a product of powers of commuting matrices and provides a loose upper bound for the scaling of the number of exponentials (which is proportional to the number of quantum operations) needed to simulate the Hamiltonian.  The theorem only considers using the construction provided in~\lem{nested}; generalizations using other product formulas (such as those generated by~\lem{genericnested}) are straightforward.

\begin{theorem}
\label{thm:simproducts}
Let $\{A_\ell \colon \ell=1,\ldots,m\}$ be commuting Hermitian matrices, let $\{\alpha_\ell \colon \ell=1,\ldots,m\}$ be integers, let $k \defeq \left(\sum_{\ell=1}^m \alpha_\ell\right) - 1$, and let $p$ be a positive integer.  Then $\exp(-i 2^k A_1^{\alpha_1} \ldots A_m^{\alpha_m} t^{k+1})$ can be simulated with error at most $\epsilon>0$ using
\begin{equation}
N_{\exp}\in O\left( \frac{6^{pk}\left({6^{pk}\Lambda t} \right)^{k+1+{(k+1)^2}/({2p})}}{\epsilon^{(k+1)/(2p)}}\right)
\end{equation}
exponentials, where $\Lambda \ge 2 \max_j \norm{A_j}$.
\end{theorem}

\begin{proof}
The proof is a simple generalization of the previous discussion.  First introduce operators $\A_j$ for $j=0,\ldots,k$ such that each $\A_j$ corresponds to one of the operators $A_\ell$, where $\A_j = A_\ell$ for $\alpha_\ell$ values of $j$.  Then define $\A_j'$ as the isometric extension
\begin{align}
  \A_j'&\defeq \begin{cases}
  \A_0\otimes \sigma_x & \text{if $j=0$}\\
  \A_j\otimes \sigma_y & \text{if $j>0$}.
  \end{cases}
\end{align}
We now show by induction that the nested commutator of all the $\A_j'$ is an isometric extension of $\A_0 \A_1 \ldots \A_k$.  Specifically, we claim that
\begin{align}
[\A_k', [ \ldots, [\A_1',\A_0']\ldots]]
&=\begin{cases}
-i2^k \prod_{q=1}^k \A_q \otimes \sigma_x & \text{$k$ odd}\\
  2^k \prod_{q=1}^k \A_q \otimes \sigma_z & \text{$k$ even} \end{cases}\nn
&=\begin{cases}
-i2^k A_0^{\alpha_0}A_1^{\alpha_1}\cdots A_k^{\alpha_k}\otimes \sigma_x & \text{$k$ odd}\\
  2^k A_0^{\alpha_0}A_1^{\alpha_1}\cdots A_k^{\alpha_k} \otimes \sigma_z & \text{$k$ even}.
\end{cases}
\end{align}

We have already demonstrated the base case in \eq{comm2} by showing that $[\A'_1,\A'_0]= -2i\A_1\A_0\otimes \sigma_z$.  Thus, assume that the claim holds for some given value of $k$.  Then we have
\begin{align}
[\A_{k+1}',[\A_k',[\ldots, [\A_1',\A_0']\ldots]]]
&=\begin{cases}
-i2^{k+1} \prod_{q=1}^{j+1}\A_q \otimes [\sigma_y,\sigma_x] & \text{$k$ odd}\\
  2^{k+1} \prod_{q=1}^{j+1}\A_q \otimes [\sigma_y,\sigma_z] & \text{$k$ even}
\end{cases} \nn
&=\begin{cases}
  2^{k+1} \prod_{q=1}^{j+1} \A_q\otimes \sigma_z & \text{$k+1$ even}\\
-i2^{k+1} \prod_{q=1}^{j+1}\A_q\otimes \sigma_x & \text{$k+1$ odd},
\end{cases}
\end{align}
and the claim follows by induction.

Finally, we simulate the nested commutator exponential by $\nestf{p}{k}(\A_k' t, \ldots, \A_0' t)$ (with an eigenstate of $\sigma_z$ or $\sigma_x$, as appropriate, in the ancilla register).  The theorem follows by using \eq{nexp} for $N_{p,k}\in O(6^{pk})$ and $Q_{p,k} < 1$ (as justified in \sec{lintime}).
\end{proof}

In the example of the toric code Hamiltonian, \thm{simproducts} implies that $e^{-iHt}$ can be simulated with error at most $\epsilon$ using
\begin{equation}
O\left(\frac{6^{3p}\left({6^{3p}(J t)^{1/4}} \right)^{4+{8}/p} n}{(\epsilon/n)^{2/p}}\right)\label{eq:toricscale}
\end{equation}
operations, where $n$ is the number of vertices (or equivalently, the number of plaquettes) in the lattice.  We replace $\epsilon$ with $\epsilon/n$ because simulation errors are subadditive and there are $O(n)$ terms to be simulated in~\eq{toricham}.   Equation~\eq{toricscale} shows that the toric code dynamics can be implemented using a number of two-body interactions that scales near-linearly with $t$ while using only one ancilla qubit (although it may be more convenient to use $O(n)$ ancilla qubits so that nearest-neighbor interactions suffice).


An alternative to \thm{simproducts}
is to estimate the eigenvalues of each $A_i$ using phase estimation and introduce a conditional phase that depends on the eigenvalue (e.g., as described in \cite{SMM08}).  
The method based on commutators
has the advantage that it needs only one ancilla qubit, rather than a logarithmic-size register used to store the estimated phase.  Also, phase estimation algorithms have the drawback of requiring $O(1/\epsilon)$ operations in general.  In contrast, our approach uses $\epsilon^{-o(1)}$ operations, so \thm{simproducts} may be useful for high-precision simulations.

\section{Conclusion}
\label{sec:conclusion}

We have presented recursive constructions that approximate exponentials of commutators as products of exponentials of the elementary terms.  We provided explicit upper bounds on the approximation error, found upper bounds on the number of elementary exponentials needed to approximate an exponential of commutators to within a fixed error tolerance, and established near-optimality of our results by relating commutator simulation to unstructured search.

Our formulas have natural applications in quantum control, where they can be used to introduce or suppress interactions.  Our formulas can be much more accurate than those used in current approaches, and consequently may have applications for designing highly accurate control sequences or dynamical decoupling sequences.  Such sequences may be valuable as initial guesses in numerical pulse finding methods.

Anticommutators can also be simulated using commutator simulation together with a dilation of the Hilbert space to switch the signs of the terms in the anticommutator.  This technique can be used to implement many-body Hamiltonians or simulate powers of Hamiltonians.  In principle, this approach could also be used to simulate an exponential of an arbitrary analytic function $f(H)$ by Taylor expanding to appropriate order and simulating a truncated series, combining our results with Lie--Trotter--Suzuki formulas~\cite{suz91}.

This work raises several natural questions.  Although our product formula approximations have nearly optimal scaling with $t$, alternative methods could come closer to or even saturate the lower bound.  However, this problem may be difficult as it has long been open in the case of exponentials of sums.  
A more tractable goal might be to seek improved performance as a function of $k$ over the formulas generated by~\lem{nested}, which would be useful since the number of exponentials needed for a high-order approximation for large $k$ can be prohibitively expensive using that approach.  Product formulas generated by~\lem{genericnested} have better scaling with $k$ but do not scale as well with $p$.  An approach that combines the best features of both constructions while still giving explicit product formulas would therefore be desirable.

Applications of our results to quantum control are reminiscent of the Solovay--Kitaev theorem~\cite{kit97}, which is also used to generate pulse sequences.  That approach is based on the group commutator, which we also use in \lem{u1}.  Our expression in~\eq{genGroupCommutator} has the advantage of more precisely implementing the desired rotation than the group commutators used in Solovay--Kitaev decompositions~\cite{DN06,KMM12,BS12}.  This suggests that our product formulas may also find application in Solovay--Kitaev algorithms.

Finally, recent work in numerical analysis \cite{BCR99,Chin10} and quantum simulation \cite{CW12} has shown that \emph{multi-product formulas}, which are linear combinations of product formulas, can provide more efficient approximations to operator exponentials.  It might be interesting to investigate whether similar techniques could lead to more efficient approximations for exponentials of commutators.


\begin{acknowledgments}
We thank David Gosset for suggesting the simulation approach presented in equation \eq{dilation}.  We also thank Troy Borneman, Christopher Granade, Thaddeus Ladd, and Seckin Sefi for useful comments and feedback.
This work was supported in part by MITACS, NSERC, the Ontario Ministry of Research and Innovation, and the US ARO/DTO.
\end{acknowledgments}


\bibliographystyle{apsrev_title}

\end{document}